\documentclass[journal]{IEEEtran}
\usepackage[utf8x]{inputenc}
\usepackage{latexsym}
\usepackage{graphicx}
\usepackage{graphics}
\usepackage[top=1in, bottom=1in, left=1in, right=1in]{geometry}
\usepackage{cite}
\usepackage{amsbsy}
\usepackage{bm}
\usepackage{amssymb}
\usepackage{amsmath}
\usepackage{amsthm}
\usepackage{algorithmic}
\usepackage{amsthm}
\usepackage{bbm}
\usepackage{graphicx}
\usepackage{epsfig}
\usepackage{color}
\usepackage{setspace}

\newcommand{\mbf}[1]{\text{\boldmath{$#1$}}}
\newcommand{\by}{{\mbf{y}}}
\newcommand{\bY}{{\mbf{Y}}}

\newcommand{\mc}[1]{\mathcal{#1}} 
\newcommand{\cP}{{\mc{P}}}
\newcommand{\cY}{{\mc{Y}}}

\newcounter{thm}
\newcounter{lem}
\newcounter{rmk}
\newcounter{prp}

\newcounter{crl}

\newtheorem{remark}[rmk]{Remark}
\newtheorem{theorem}[thm]{Theorem}
\newtheorem{lemma}[lem]{Lemma}
\newtheorem{proposition}[prp]{Proposition}
\newtheorem{corollary}[crl]{Corollary}

\ifCLASSINFOpdf
  \else 
  \fi
  
\definecolor{mred}{rgb}{0,0,0}
\definecolor{dgreen}{rgb}{0,0,0}
\definecolor{mgreen}{rgb}{0,0,0}
\definecolor{dblue}{rgb}{0,0,0}
\definecolor{dbblue}{rgb}{0,0,0}
\definecolor{mblue}{rgb}{0,0,0}
\definecolor{dyellow}{rgb}{0.5,0.5,0}
\definecolor{dmagenta}{rgb}{0.75,0,0.75}
\definecolor{dcyan}{rgb}{0,0.5,0.5}
\def\tmred{\textcolor{mred}}

\def\tmblue{\textcolor{dblue}}
\def\tdblue{\textcolor{dbblue}}
\def\tdbblue{\textcolor{dbblue}}

\begin{document}

\title{Universal Outlier Hypothesis Testing}

\author{
Yun Li, \IEEEmembership{Student Member,~IEEE,} Sirin Nitinawarat, \IEEEmembership{Member,~IEEE,} 
and\\ Venugopal V. Veeravalli, \IEEEmembership{Fellow,~IEEE}
\thanks{
This work was supported by the Air Force Office of
Scientific Research (AFOSR) under the Grant
FA9550-10-1-0458 through the University of Illinois at
Urbana-Champaign, by the U.S. Defense Threat Reduction
Agency through subcontract 147755 at the University of
Illinois from prime award HDTRA1-10-1-0086, and by the
National Science Foundation under Grant NSF CCF
11-11342.  The material in this paper was
presented in part at the IEEE International Symposium on
Information Theory, Istanbul, Turkey June 7-12, 2013.
}
\thanks{The authors are with the Department of Electrical and Computer Engineering and the Coordinated Science Laboratory, University of Illinois, Urbana, IL 61801.
Email: $\left \{ \right.$\texttt{yunli2,nitinawa,vvv}$\left.
\right \}$\texttt{@illinois.edu.}}
\thanks{Copyright (c) 2013 IEEE.}
}
\maketitle
\begin{abstract}
Outlier hypothesis testing is studied in a universal setting. 
Multiple sequences of observations are collected, a small subset of which are outliers.
A sequence is considered an outlier if the observations in that sequence are distributed according to an ``outlier'' distribution, distinct from the ``typical'' distribution governing the observations in all the other sequences. 
Nothing is known about the outlier and typical distributions except that they are distinct and have full supports.  The goal is to design a universal test to best discern the \tmblue{outlier sequence(s)}.  For models with exactly one \tmblue{outlier sequence, the generalized likelihood test is shown to be universally exponentially consistent.} A single-letter characterization of the error exponent achievable by the test is derived, and  it is shown that the test 
achieves the optimal error exponent asymptotically as the number of sequences approaches infinity.
When the null hypothesis with no outlier is included, a modification of \tmblue{the generalized likelihood test} is shown to achieve the same error exponent under each non-null hypothesis, and also consistency under the null hypothesis. Then, models with more than one outlier are studied in the following settings. For the setting with a known number of distinctly distributed outliers, \tmblue{the achievable error exponent of the generalized likelihood test is characterized.} The limiting error exponent achieved by such a test is characterized, and the test is shown to be asymptotically exponentially consistent. For the setting with an unknown number of identically distributed outliers, \tmblue{a modification of the generalized likelihood test} is shown to achieve a positive error exponent under each non-null hypothesis, and also consistency under the null hypothesis. When the \tmblue{outlier sequences} can be distinctly distributed (with their total number being unknown), it is shown that a universally exponentially consistent test cannot exist, even when the typical distribution is known and the null hypothesis is excluded.
\end{abstract}


\section{Introduction}
\label{sec-intro}

We consider the following inference problem, which we term {\em outlier hypothesis testing.} 
Among a large number, say $M$, of observation sequences, it is assumed that there is a small subset of outlier sequences. Specifically, when the $i$-th sequence is an outlier, the distribution governing the observations in that sequence is assumed to be distinct from that governing the observations in all  the ``typical'' sequences.
The goal is to design a test to identify  {\em all} the \tmblue{outlier sequences}. We shall be interested in a {\em universal} setting of this problem, where the test has to perform well without any prior knowledge of the outlier and typical distributions.

It is to be noted that outlier hypothesis testing is distinct from statistical {\em outlier detection} \cite{barn-applstat-1978, hawk-iden-outlier-book-1980}.  In outlier detection, the goal is to efficiently winnow out a few outlier observations from a single sequence of observations. The outlier observations are assumed to follow a different generating mechanism from that governing the normal observations. The main differences between this outlier detection problem and outlier hypothesis testing are: (i) in the former problem, the outlier observations constitute a much smaller fraction of the entire observations than in the latter problem, and (ii) these outlier observations can be arbitrarily spread out among all observations in the outlier detection problem, whereas all the outlier observations are concentrated in a {\em fixed} subset of \tmblue{sequences} in the outlier hypothesis testing problem.  

Statistical outlier detection is typically used to preprocess large data sets, to obtain clean data that is used for some purpose such as inference and control. The \tmblue{outlier hypothesis testing problem that we study here arises in fraud and anomaly detection in large data sets \cite{bolt-hand-2002, chan-bane-kuma-2009}, severe weather prediction, environment monitoring in sensor networks \cite{cham-veer-2007}, network intrusion and voting irregularity analysis. It also finds applications where the outlying sequences have a more positive connotation, such as spectrum sensing and high frequency trading.} 

Universal outlier hypothesis testing is related to a broader class of {\em composite hypothesis testing} problems in which there is uncertainty in the probabilistic laws associated with some or all of the hypotheses. To solve these problems, a popular approach \tmblue{is to apply the {\em generalized likelihood (GL) test}} \cite{poor-detest-book-1994, zeit-ziv-merh-1992}. For example, in the {\em simple-versus-composite} case, the goal is to make a decision in favor of either the null distribution, which is known to the tester, or a family of alternative distributions. A fundamental result concerning the asymptotic optimality of the \tmblue{{\em generalized likelihood ratio test (GLRT)}} in this case was shown in \cite{hoef-amstat-1965}. When some uncertainty is present in the null distribution as well, i.e., the {\em composite-versus-composite} setting, the optimality of \tmblue{the GLRT} has been examined under various conditions \cite{zeit-ziv-merh-1992}.

Universal outlier hypothesis testing is \tdblue{also} closely related to homogeneity testing and classification \cite{pear-1911, shay-2011, unni-2012, ziv-1988, gutm-1989}.
In homogeneity testing, one wishes to decide whether or not \tmblue{two samples} come from the same probabilistic law. In classification problems, a set of test data is classified to one of 
multiple streams of training data with distinct labels. 
Metrics that are commonly used to quantify the performance of a test are {\em consistency} and {\em exponential consistency}. A universal test is {\em consistent}  if the error probability approaches zero as the sample size goes to infinity, and is {\em exponentially consistent} if the decay is exponential with sample size.
In \cite{ziv-1988, gutm-1989}, a classifier based on the principle of the \tmblue{GL test} was shown to be optimal under the asymptotic Neyman-Pearson criterion. In particular, in \cite{ziv-1988}, the classifier is designed to minimize the error probability under the \tmblue{inhomogeneous hypothesis}, under a predefined constraint on the exponent for the error probability under the \tmblue{homogeneous hypothesis}. And,  in \cite{gutm-1989}, the classifier is designed to minimize the probability of rejection, under a constraint on the probability of misclassification. \tmblue{However, the aforementioned optimality is achieved only when the length of the training data grows at least linearly with that of the test data, and the distribution of the test data is separated enough from those of all unmatched training data.}


In universal outlier hypothesis testing, we have no information regarding the outlier and typical distributions. In particular, the outliers can be arbitrarily distributed as long as each of 
them is distinctly distributed from the typical distribution. In addition, we have no training data to learn these distributions before the detection is performed. 
As a consequence,
it is not clear at the outset that a universally exponentially consistent test should exist, and even if it does, it is not clear what its structure and performance should be. \tmblue{Quite surprisingly, we are able to show that even with {\em no} training data, the GL test is far more efficient for the outlier hypothesis testing than for the other inference problems mentioned previously, such as homogeneity testing and classification.} 

\tmblue{Universal outlier hypothesis testing is related to universal coding over discrete memoryless channels, e.g., \cite{fede-lapi-1998}.  Nevertheless, there is an important distinction as follows.  In universal coding, the encoder and decoder are jointly optimized to achieve universality.  On the other hand, in outlier hypothesis testing, when properly interpreted, only the decoding is sought to be optimized, while the encoding scheme is fixed by the structure of the distributions of observations among all hypotheses, and cannot be chosen.  Consequently, the results in \cite{fede-lapi-1998} cannot be applied to our problem (see also the remark after Example 1 for more details).}

Our technical contributions are as follows.  In Section \ref{sec-single}, we consider models with at most one outlier.  
When the outlier is always present and the typical distribution is known, we show that \tmblue{the GL test} achieves the same optimal error exponent as if the outlier distribution is known as well.  We also show (in Example 1) that the same optimal error exponent cannot be achieved universally by {\em any} test when the typical distribution is not known.  In such a completely universal setting, we prove that \tmblue{the GL test} is {\em exponentially consistent} for each $M \geq 3$.  We also  establish that as $M$ goes to infinity, the error exponent achievable by \tmblue{the GL test} converges to the optimal error exponent corresponding to the case where both the typical and outlier distributions are known.  When there is also the {\em null} hypothesis with no outlier, we show that there cannot exist a universally exponentially consistent test even when the typical distribution is known.  Nevertheless, by slightly modifying \tmblue{the GL test}, we are able to achieve the same error exponent under every hypothesis with the outlier present, and consistency under the null hypothesis.  In Section \ref{sec-multiple}, we consider models with multiple outliers. For models with a known number of distinctly distributed outliers, \tmblue{the GL test is shown to achieve} a positive limiting error exponent as $M$ goes to infinity,  whenever the limit of the optimal error exponent, when all the outlier and typical distributions are known, is positive.  
For models with an unknown number of {\em identical} outliers, \tmblue{a positive error exponent is achieved by a similar modification of the GL test} under every hypothesis with outliers, and consistency under the null hypothesis universally.  For models with unknown number of {\em distinct} outliers, we show that there cannot exist a universally exponentially consistent test, even when the typical distribution is known and when the null hypothesis is excluded.

\section{Preliminaries} \label{sec-prelim}

Throughout the paper, random variables are denoted by capital letters, and their realizations are denoted by the corresponding lower-case letters. All random variables are assumed to take values in finite sets, and all logarithms are the natural ones.
 
For a finite set $\cY$, let $\cY^m$ denote the $m$ Cartesian product of $\cY$, and $\mathcal{P}(\cY)$ denote the set of all probability mass functions (pmfs) on $\cY$. The empirical distribution of a sequence $\by=y^m = \left( y_1, \ldots, y_m \right) \in \cY^m,$ denoted by $\gamma = \gamma_{\by} \in \mathcal{P} (\cY ),$ is defined as
\begin{align}
\gamma(y) \ \triangleq \ \frac{1}{m} \big \vert \left \{ k=1, \ldots, m : y_{k}=y \right \} \big \vert, \nonumber
\end{align}
$y \in \cY$.

Our results will be stated in terms of various distance metrics between a pair of distribution $p, q \in \cP \left( \cY \right).$  In particular, we shall consider two symmetric distance metrics: the {\em Bhattacharyya distance} and {\em Chernoff information}, denoted respectively by $B (p, q)$ and $C(p, q),$ and defined as (see, e.g., \cite{cove-thom-eit-book-2006})
\begin{align} 
B(p, q) \ \triangleq \ 
-\log 
\left ( 
\sum_{y \in \cY} p(y)^{\frac{1}{2}}  q(y)^{\frac{1}{2}} 
\right )
\label{def-Bhat}
\end{align}
and 
\begin{align} 
C(p, q) \ \triangleq\  
\max_{s \in [0, 1]} 
-\log \left ( \sum_{y \in \cY} p(y)^{s} q(y)^{1-s} \right ), 
\label{def-Cher}
\end{align}
respectively.
Another distance metric, which will be key to our study, is the relative entropy, denoted by $D ( p \| q )$ and defined as
\begin{align}
D(p \| q) \ \triangleq \ \sum_{y \in \cY} p(y) \log \frac{p(y)}{q(y)}.
\label{def-KL}
\end{align}
Unlike the Bhattacharyya distance (\ref{def-Bhat}) and Chernoff information (\ref{def-Cher}), the relative entropy in (\ref{def-KL}) is a {\em non-symmetric} distance \cite{cove-thom-eit-book-2006}.

The following technical facts will be useful; their derivations can be found in \tdbblue{\cite[Theorem 11.1.2]{cove-thom-eit-book-2006}.} Consider random variables $Y^{n}$ which are i.i.d. according to $p \in \mathcal{P} (\cY)$.  
Let $y^{n} \in \cY^{n}$ be a sequence with an empirical distribution $\gamma \in \mathcal{P}(\cY)$.  It follows that the probability of such sequence $y^{n},$ under $p$ and under the i.i.d. assumption, is 
\begin{align} 
p(y^{n}) \ =\ \exp \Big \{-n \, \big ( D(\gamma \| p) + H(\gamma) \big ) \Big \}, 
\label{eqn-probability}
\end{align}
where $D(\gamma \| p)$ and $H(\gamma)$ are the relative entropy of $\gamma$ and $p$, and entropy of $\gamma$, defined as
\begin{align}
D(\gamma \| p) \ \triangleq \ \sum_{y \in \cY} \gamma(y) \log \frac{\gamma(y)}{p(y)}, \nonumber
\end{align}
and
\begin{align}
H(\gamma) \ \triangleq \ - \sum_{y \in \cY} \gamma(y) \log \gamma(y), \nonumber
\end{align}
respectively.
Consequently, it holds that for each $y^n$, the pmf $p$ that maximizes $\,p(y^{n})\,$ is $\,p=\gamma$, and the associated maximal probability of $y^{n}$ is 
\begin{align} \label{eq-ML}
\gamma(y^{n}) \ = \ \exp \big \{ -nH(\gamma) \big \}.
\end{align}

\section{\tmblue{Universal Outlier Hypothesis Testing}} \label{sec-single}

\subsection{Models with Exactly One Outlier} \label{sec-UOHT-exactone}
\tmblue{Consider $M \geq 3$ independent sequences of observations, each of which consists of $n$ independent and identically distributed (i.i.d.) observations. We denote the $k$-th observation of the $i$-th sequence by $Y_{k}^{(i)} \in \cY$. It is assumed that only one sequence is the ``outlier,'' i.e., the observations in that sequence are uniquely distributed (i.i.d.) according to the ``outlier'' distribution $\mu \in \mathcal{P}(\cY)$, while all the other \tmblue{sequences} are commonly distributed according to the ``typical'' distribution $\pi \in \mathcal{P}(\cY)$. } 
{\em Nothing is known about $\mu$ and $\pi$ except that $\mu \neq \pi$, and that each of them has a full support.} \tmblue{The assumption of $\mu, \pi$ having full support rules out trivial cases where it is straightforward to identify} the \tmblue{outlier sequence.} Clearly, if $M=2$, either \tmblue{sequence} can be considered as an outlier; hence, it becomes degenerate to consider outlier hypothesis testing in this case.

\tmblue{It is assumed throughout this section that the outlier distribution is} 
unknown but is independent of the identity of the outlier. 
\tmblue{In certain applications, it may be natural to consider the model where the outlier distribution can vary depending on the identity of the outlier. This scenario can be viewed as a special case (with the number of outlier sequences being exactly one) of the multiple outlier hypothesis testing problem studied in Section \ref{sec-multiple}}.

\tmblue{Conditioned on the hypothesis that the $i$-th sequence is the outlier}, the joint distribution of all the observations is
\begin{align}
p_{i} \big ( y^{Mn} \big ) & \ = \
p_{i} \left ( \by^{(1)}, \ldots, \by^{(M)} \right ) \nonumber \\
&\ = \ \prod_{k=1}^{n} \Big \{ \mu \big ( y_{k}^{(i)} \big ) \, \prod_{j \neq i} \pi \big ( y_{k}^{(j)} \big ) \Big \}
\nonumber \\
&\ \triangleq\   \tmblue{L_{i} \left (y^{Mn}, \mu, \, \pi \right ),} 
\label{eqn-jointdis}
\end{align}
where
\begin{align}
\by^{(i)} \ = \ \left ( y_{1}^{(i)}, \ldots, y_{n}^{(i)} \right ),\ i = 1, \ldots, M. \nonumber 
\end{align}
The test for the \tmblue{outlier sequence} is done based on a {\em universal} rule $\delta : \cY^{Mn} \rightarrow \{1, \ldots, M\}$. In particular, the test $\delta$ is not allowed to \tmblue{be a function of} $(\mu, \pi)$. 

For a universal test, the maximal error probability, which will be a function of the test and $(\mu, \pi)$, is
\begin{align} 
e \big (\delta, (\mu, \pi) \big) \ \triangleq \  \max_{i = 1, \ldots , M} \  
\sum_{y^{Mn}:\ \delta(y^{Mn}) ~\neq~ i}
p_{i} \big ( y^{Mn} \big ), \nonumber
\end{align}
and the corresponding error exponent is defined as
\begin{align}
\alpha \big ( \delta, (\mu,\pi) \big ) \ \triangleq \ \lim_{ n \rightarrow \infty}  -\frac{1}{n} \log e \big (\delta, (\mu, \pi) \big ). \label{eqn-errexpdef}
\end{align} 
Throughout the paper, we consider the error exponent as $n$ goes to infinity, while $M$, and hence the number of hypotheses, is kept fixed.  Consequently, the error exponent in (\ref{eqn-errexpdef}) also coincides with the one for the average probability of error.

A test is termed {\em universally consistent} if the maximal error probability converges to zero as the number of samples goes to infinity, i.e.,
\begin{align}
e \big (\delta, (\mu, \pi) \big) \rightarrow 0, \label{eqn-def-consistency}
\end{align}
for any $(\mu, \pi)$, $\mu \neq \pi$ as $n \rightarrow \infty$.
It is termed {\em universally exponentially consistent} if the exponent for the maximal error probability is strictly positive, i.e.,
\begin{align}
\alpha \big ( \delta, (\mu,\pi) \big ) > 0, \label{eqn-def-exp-consistency}
\end{align}
for any $(\mu, \pi)$, $\mu \neq \pi$.

\subsubsection{\tmblue{Generalized Likelihood Test}} \label{sec-UOHT-exactone-test}
We now describe the \tmblue{generalized likelihood (GL) test} in two setups when only $\pi$ is known, and when neither $\mu$ nor $\pi$ is known, respectively.  


For each $i = 1, \ldots, M$, denote the empirical distributions of $\by^{(i)}$ by $\gamma_{i}$.
When $\pi$ is known and $\mu$ is unknown, \tmblue{conditioned on the $i$-th sequence being the outlier, $i = 1, \ldots, M,$ we compute the generalized likelihood of $y^{Mn}$ by replacing $\mu$ in (\ref{eqn-jointdis}) with its maximum likelihood (ML) estimate $\hat{\mu}_i \triangleq \gamma_i,$ as}
\begin{align}
\tmblue{\hat{p}^{\mbox{\scriptsize{typ}}}_{i} \left( y^{Mn} \right)  \ = \ L_{i} \left (y^{Mn}, \hat{\mu}_{i}, \, \pi \right ).} \label{eqn-detector-pi-known}
\end{align}

Similarly, when neither $\mu$ nor $\pi$ is known, we compute the generalized likelihood of $y^{Mn}$ by replacing the $\mu$ and $\pi$ in (\ref{eqn-jointdis}) with their ML estimates $\hat{\mu}_i \triangleq \gamma_i$, and $\hat{\pi}_i \triangleq \textstyle{\frac{\sum_{k \neq i} \gamma_{k}}{M-1}},\ i = 1, \ldots, M$, as
\begin{align}
\tmblue{\hat{p}^{\mbox{\scriptsize{univ}}}_{i} \left( y^{Mn} \right)  \ = \ L_{i} \left (y^{Mn}, \hat{\mu}_{i}, \, \hat{\pi}_{i} \right ).} \label{eqn-univ-detector}
\end{align}

Finally, we decide upon the \tmblue{sequence corresponding to the largest generalized likelihood to be the outlier.}  Using (\ref{eqn-detector-pi-known}), (\ref{eqn-univ-detector}), \tmblue{the GL tests} in the two cases can be described respectively as
\begin{align}
\tmblue{ \delta \big (y^{Mn} \big ) \ = \ \mathop{\mbox{argmax}}\limits_{i=1, \ldots, M} \ \hat{p}^{\mbox{\scriptsize{typ}}}_{i} \left( y^{Mn} \right)} \label{eqn-detector-pi-known-2}
\end{align}
\tmblue{when only $\pi$ is known, and}
\begin{align}
\tmblue{ \delta \big (y^{Mn} \big ) \ = \ \mathop{\mbox{argmax}}\limits_{i=1, \ldots, M} \ \hat{p}^{\mbox{\scriptsize{univ}}}_{i} \left( y^{Mn} \right)} \label{eqn-univ-detector-2}
\end{align}
\tmblue{when neither $\mu$ nor $\pi$ is known.}
In (\ref{eqn-detector-pi-known-2}) and (\ref{eqn-univ-detector-2}), should there be multiple \tmblue{maximizers}, we pick one of them arbitrarily. \tmblue{Using the identity in (\ref{eqn-probability}), it is straightforward to show that when only $\pi$ is known, the GL test in (\ref{eqn-detector-pi-known-2}) is equivalent to}
\begin{align}
\delta \big (y^{Mn} \big )
&\, =\, \mathop{\mbox{argmin}}\limits_{i=1, \ldots, M}\  
H \left( \gamma_i \right) 
+ \sum\limits_{j \neq i}
\left [ 
H \left( \gamma_j \right) + D \left( \gamma_j \| \pi \right)
\right ]
\nonumber \\
&\, = \, \mathop{\mbox{argmax}}\limits_{i=1, \ldots, M} \ D(\gamma_{i} \| \pi),	
\label{eqn-detector-pi-known-2'}
\end{align}
\tmblue{and when neither $\pi$ nor $\mu$ is known, the test in (\ref{eqn-univ-detector-2}) is equivalent to}
\begin{align}
\hspace{-0.05in} \delta \big (y^{Mn} \big ) 
\hspace{-0.02in} & = \mathop{\mbox{argmin}}\limits_{i=1, \ldots, M} 
H \left( \gamma_i \right) \nonumber  \\
& \hspace{0.2in} + \sum\limits_{j \neq i}
\left [ 
H \left( \gamma_j \right) + 
D \big (\gamma_{j} \big \| \textstyle\frac{\sum_{k \neq i} \gamma_{k}}{M-1} \big )
\right ]
\nonumber \\
&=
\mathop{\mbox{argmin}}\limits_{i=1, \ldots, M} 
\ \sum_{j \neq i} D \big (\gamma_{j} \big \| \textstyle\frac{\sum_{k \neq i} \gamma_{k}}{M-1} \big ).
\label{eqn-univ-detector-2'}
\end{align}

\subsubsection{Results}
\label{sec-UOHT-exactone-results}

Our first theorem \tdblue{for models with one outlier} characterizes the optimal exponent for the maximal error probability when both $\mu$ and $\pi$ are known, and when only $\pi$ is known.
\begin{theorem} \label{thm-1}
When $\mu$ and $\pi$ are both known, the optimal exponent for the maximal error probability is equal to
\begin{align}
2 B(\mu, \pi).  \label{eqn-completely-nonuniv}
\end{align}
\tmblue{Furthermore, the error exponent in (\ref{eqn-completely-nonuniv}) is achievable by the GL test in (\ref{eqn-detector-pi-known-2}), which uses only the knowledge of $\pi$.} 
\end{theorem}

\begin{remark} \label{rmk-1} {\em
It is interesting to note that when only $\mu$ is known, one can also achieve the optimal error exponent in (\ref{eqn-completely-nonuniv}) using a different test that will be presented in Appendix \ref{app-B}.
However, we do not yet know if the corresponding version of the GL test, wherein the $\pi$ in (\ref{eqn-jointdis}) is replaced with $\hat{\pi}_i = \textstyle{\frac{\sum_{k \neq i} \gamma_{k}}{M-1}},\ i = 1, \ldots, M,$ is optimal.}  
\end{remark}

Consequently, in the completely universal setting, when nothing is known about $\mu$ and $\pi$ except that $\mu \neq \pi$, and both $\mu$ and $\pi$ have full supports, it holds that for any universal test $\delta$,
\begin{align}
\alpha \big ( \delta, ( \mu, \pi) \big ) \ \leq \ \ 2B(\mu, \pi). \label{eqn-uppbd}
\end{align}

\tdblue{Given the second assertion in Theorem \ref{thm-1}, it might be tempting to think that it would be possible to design a test to achieve the optimal error exponent of $2 B \left( \mu, \pi \right)$ universally when neither $\mu$ nor $\pi$ is known.  Our first example shows that such a goal cannot be fulfilled, and hence we need to be content with a lesser goal.}

\vspace{0.1in}
\noindent
\textbf{Example 1:}  Consider the model with $M = 3,$ and a distinct pair of distributions $p \neq \overline{p}$ on $\cal{Y}$ with full supports.  We now show that there cannot exist a universal test that achieves the optimal error exponent of $2B \left( \mu, \pi \right)$ {\em even just for the two models} when $\mu = p, \pi = \overline{p},$ and when $\mu = \overline{p}, \pi = p,$ both of which have $2B \left( \mu, \pi \right) = 2 B \left( p, \overline{p} \right).$  To this end, let us look at the region when a universal test $\delta$ decides that the first \tmblue{sequence} is the outlier, i.e., $A_1 = \left \{ y^{3n}:\ \delta \left( y^{3n} \right) = 1 \right \}$.  Let $\mathbb{P}_{p, \overline{p}, \overline{p}}$ denote the distribution corresponding to the first hypothesis of the first model, i.e., when $\by^{(1)}$ are i.i.d. according to $p,$ and $\by^{(2)}$ and $\by^{(3)}$ are i.i.d. according to $\overline{p}.$  Similarly, let $\mathbb{P}_{p, \overline{p}, p}$ denote the distribution corresponding to the second hypothesis of the second model, i.e., when $\by^{(2)}$ are i.i.d. according to $\overline{p},$ and $\by^{(1)}$ and $\by^{(3)}$ are i.i.d. according to $p.$  Suppose that $\delta$ achieves the best error exponent of $2 B \left( p, \overline{p} \right)$ for the first model when $\mu = p, \pi = \overline{p}.$  Then, it must holds that
\begin{align}
\lim_{n \rightarrow \infty}\ 
-\frac{1}{n} \log{
\mathbb{P}_{p, \overline{p}, \overline{p}} \left \{ A_1^c \right \}
} 
\ \geq\ 
2 B \left( p, \overline{p} \right).
\label{example1-eq1}
\end{align}
\tmred{It now follows from (\ref{example1-eq1}) and the classic result of Hoeffding \cite{hoef-amstat-1965} in binary hypothesis testing (see, e.g., \cite{csis-korn-book-2011}[Excercise 2.13 (b)]) that
\begin{align}
& \lim_{n \rightarrow \infty}\ 
-\frac{1}{n} \log{
\mathbb{P}_{p, \overline{p}, p} \left \{ A_1 \right \}
} \nonumber
\\ 
& \leq \Big [ 
\mathop{\min\limits_{
{ q\left( y_1, y_2, y_3 \right)}
}}
D\left( q\left( y_1 \right) \| p \right) + 
D\left( q\left( y_2 \right) \| \overline{p} \right)   \nonumber \\
&  \hspace{1.57in}+  D\left( q\left( y_3 \right) \| p \right) 
\Big  ] ^{+}
\nonumber \\
& \leq \Big [
\min\limits_{
{\small q\left( y_1, y_2, y_3 \right)}}
2 B \left( p, \overline{p} \right)  +
D\left( q\left( y_3 \right) \| p \right) \nonumber \\
& \hspace{1.35in} - D\left( q\left( y_3 \right) \| \overline{p} \right) 
 \Big ]^{+}
\nonumber \\
& \leq
\big (
2 B \left( p, \overline{p} \right) -
D\left( p \| \overline{p} \right) 
\big )^{+}  \ =\ 0,
\label{example1-eq2}
\end{align}
where each minimum on the right-side above is
taken over the set of $q(y_{1}, y_{2}, y_{3})$ such that
\begin{align}
\hspace{-0.02in} D\left( q\left( y_1 \right) \| p \right) \hspace{-0.02in}+\hspace{-0.02in} 
D\left( q\left( y_2 \right) \| \overline{p} \right)\hspace{-0.02in} + \hspace{-0.02in}
D\left( q\left( y_3 \right) \| \overline{p} \right) 
\leq 2 B \left( p, \overline{p} \right). \nonumber
\end{align}
The last equality in (\ref{example1-eq2}) follows from Lemma \ref{lm-2} in Appendix \ref{sec-app1}.  Consequently, the test cannot yield even a {\em positive} error exponent for the second model when $\mu = \overline{p}, \pi = p.$}

\begin{remark}

\tmblue{{\em
It is interesting to contrast this example for outlier hypothesis testing with the results (Theorem 2 and 3 in \cite{fede-lapi-1998}) for universal coding over discrete memoryless channels (DMCs).
Specifically, Theorem 2 and 3 in \cite{fede-lapi-1998} establish that the optimal error exponent at zero rate is universally achieved for all DMCs, whereas the optimal error exponent $2B \left( \mu, \pi \right)$ for outlier hypothesis testing here {\em cannot} be universally achieved.  The difference between these two results stems from the following distinctions between the nature of these two problems.  First, in universal coding, the encoder and decoder are jointly optimized to achieve universality.  On the other hand, in outlier hypothesis testing, when properly interpreted, only the decoding is allowed to be optimized, while the encoding scheme is fixed by the structure of the distributions of observations among all hypotheses, and cannot be chosen.  Second,  the zero-rate error exponent in \cite{fede-lapi-1998} applies only for the case when the number of messages {\em grows} to infinity with the blocklength sub-exponentially.  In contrast, the number of hypotheses in outlier hypothesis testing is fixed and does not grow with the number of observations in each sequence.}}

\tmblue{{\em To summarize, the results in \cite{fede-lapi-1998} cannot be applied to our problem. Had the results in \cite{fede-lapi-1998} been applicable, Theorems 2 and 3 in \cite{fede-lapi-1998} would have implied that the optimal error exponent $2B \left( \mu, \pi \right)$ is achieved universally for outlier hypothesis testing as well. However, Example 1 proves otherwise.}}
  
\end{remark}

Example 1 shows explicitly that when neither $\mu$ nor $\pi$ is known, it is impossible to construct a test that achieves $2 B \left( \mu, \pi \right)$ universally.  In fact, the example shows that had we insisted on achieving the best error exponent of $2 B \left( \mu, \pi \right)$ for some pairs of $\mu, \pi,$ it might not be possible to achieve even {\em positive} error exponents for some other pairs of $\mu, \pi$.  This motivates us to seek instead a test that yields just a positive (no matter how small) error exponent $\alpha \left( \delta, (\mu, \pi) \right) > 0$ for {\em every} $\mu, \pi$, $\mu \neq \pi,$ i.e., achieving universally exponential consistency.   
One of our main contributions in this paper is to show that GL tests are indeed {\em universally exponentially consistent} under various settings, including the current single outlier setting for every {\em fixed} $M$.
\begin{theorem} \label{thm-2}
The GL test $\delta$ in (\ref{eqn-univ-detector-2}) is universally exponentially consistent. Furthermore, for every pair of distributions $\mu, \pi, \mu \neq \pi$, it holds that
\begin{align} 
\alpha \big ( \delta, (\mu, \pi) \big )  =
\min\limits_{q_{1}, \ldots, q_{{\scriptsize M}}} 
& D \left( q_{1} \| \mu \right) +  D \left( q_{2} \| \pi \right) \nonumber \\
& \hspace{0.2in}  + \ldots  +  D\left ( q_{{\scriptsize M}} \| \pi \right),  \label{eqn-err-exp}
\end{align}
where the minimum above is over the set of $\left( q_1, \ldots, q_{\scriptsize M} \right)$ such that
\begin{align} 
\sum_{j \neq 1} D \left ( q_{j} \, \Big \| \, \textstyle{\frac{ \sum_{k \neq 1} q_{k}}{M-1}} \right ) 
\geq
\sum_{j \neq 2} D \left ( q_{j} \, \Big \| \, \textstyle{\frac{\sum_{k \neq 2} q_{k}}{M-1}} \right ).  \label{eqn-constraint}
\end{align}
\end{theorem}

Note that for any fixed $M \geq 3$, $\epsilon > 0$, regardless of which \tmblue{sequence} is the outlier, it holds that the random empirical distributions $\left (\gamma_{1}, \ldots, \gamma_{M} \right )$ satisfy
\begin{align}
\lim_{n \rightarrow \infty} \mathbb{P}_{i} \left \{ \ \left \| \textstyle{\frac{1}{M} \sum_{j=1}^{M} \gamma_{j} - \left (\frac{1}{M}\mu + \frac{M-1}{M}\pi \right )} \right \|_{1} > \epsilon  \right \} = 0,
\end{align}
where $\| \cdot \|_1$ denotes the 1-norm of the argument distribution.
Since $\frac{1}{M}\mu + \frac{M-1}{M} \pi \rightarrow \pi$ as $M \rightarrow \infty$, heuristically speaking, a consistent estimate of the typical distribution can readily be obtained asymptotically in $M$ from the entire observations before deciding upon which \tmblue{sequence} is the outlier. This observation and the second assertion of Theorem \ref{thm-1} motivate our study of the asymptotic performance (achievable error exponent) of \tmblue{the GL test in (\ref{eqn-univ-detector-2})} when $M \rightarrow \infty$ \tdblue{(after having taken the limit as $n$ goes to infinity first).} 

Our last result for models with one outlier shows that in the completely universal setting, as $M \to \infty$, \tmblue{the GL test in (\ref{eqn-univ-detector-2})} achieves the optimal error exponent in (\ref{eqn-completely-nonuniv}) corresponding to  the case in which {\em both $\mu$ and $\pi$ are known}.
 
\begin{theorem} \label{thm-3}
\tdbblue{For each $M \geq 3$,} the exponent for the maximal error probability achievable by \tmblue{the GL test $\delta$ in (\ref{eqn-univ-detector-2})} is lower bounded by 
\begin{align}
\mathop{\min\limits_{q \, \in \, \mathcal{P}(\cY)}}_{D(q \| \pi) \, \leq \, \frac{1}{M-1} \big (2B(\mu, \pi) + C_{\pi} \big )}  2\, B (\mu \, , \, q )\, , \label{eqn-lowerbd}
\end{align}
where $C_{\pi} \triangleq -\log \Big ( \min\limits_{y \in \cY} \, \pi (y) \Big ) < \infty$ by the fact that $\pi$ has a full support.

The lower bound for the error exponent in (\ref{eqn-lowerbd}) is nondecreasing in $M \geq 3$.
Furthermore, as $M \rightarrow \infty$, this lower bound converges to the optimal error exponent $2B ( \mu, \pi )$; hence, \tmblue{the GL test $\delta$ in (\ref{eqn-univ-detector-2})} 
achieves the optimal error exponent asymptotically as the number of sequences approaches infinity, i.e.,
\begin{equation}
\lim_{M \rightarrow \infty} \alpha \big ( \delta, ( \mu, \pi ) \big ) \ = \  2 B( \mu, \pi),   \label{eqn-limit}
\end{equation}
which from Theorem \ref{thm-1} is equal to the optimal error exponent when both $\mu$ and $\pi$ are known.
\end{theorem}

\noindent
\textbf{\tdblue{Example 2:}} 
We now provide some numerical results for an example with $\cY=\{0, 1\}$.  Specifically, the three plots in the figure below are for \tmred{three pairs of} outlier and typical distributions being $\mu=(p(0)=0.3,\ p(1)=0.7),\ \pi=(0.7, 0.3);\ \mu=(0.35, 0.65),\ \pi=(0.65, 0.35);$ and 
$\mu=(0.4, 0.6),\ \pi=(0.6, 0.4),$ respectively.  Each horizontal line corresponds to 
$2B ( \mu, \pi )$, and each curve line corresponds to the lower bound in (\ref{eqn-lowerbd}) for the error exponent achievable by \tmblue{the GL test in  (\ref{eqn-univ-detector-2})}.  As shown in these plots, the lower bounds converge to $2B(\mu, \pi)$ as $M \rightarrow \infty$, i.e., \tmblue{the GL test in  (\ref{eqn-univ-detector-2})} is asymptotically optimal for all three pairs \tmblue{of} $\mu, \pi$, and, indeed, for all $\mu \neq \pi$.

\begin{figure}[h] \label{fig-1}
\hspace{-0.2in}
\epsfig{file=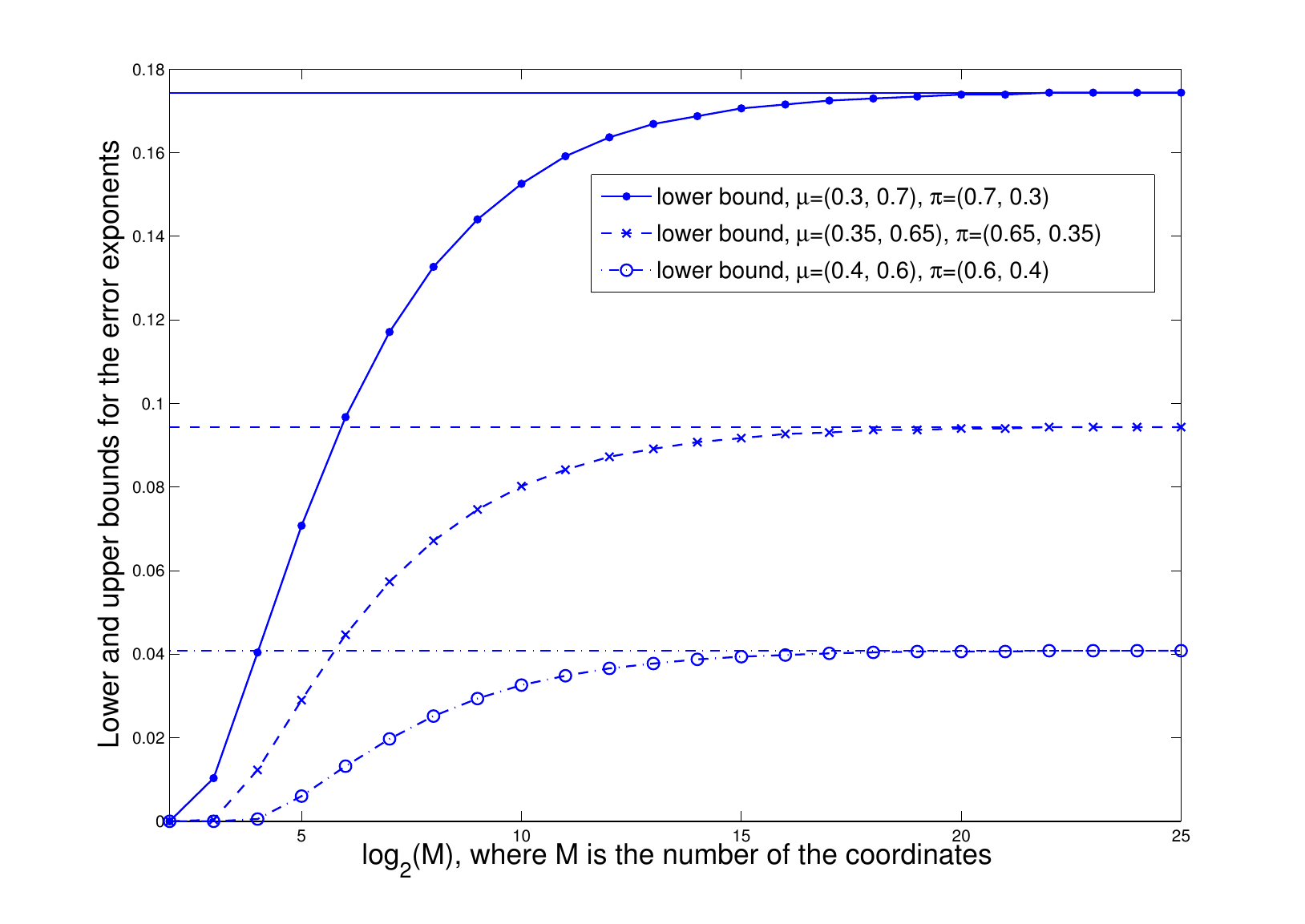, height = 2.6 in, width=3.55 in}
\end{figure}

\subsection{\tmblue{Models with At Most One Outlier Sequence}}
A natural question that arises at this point is what would happen if it is also possible that no outlier is present.  To answer this question, we now consider models that append an additional {\em null} hypothesis with no outlier to the previous models consider in Section \ref{sec-UOHT-exactone}.  In particular, under the null hypothesis, the joint distribution of all the observations is given by
\[
\begin{split}
p_{0} \big ( y^{Mn} \big ) \ = \ \prod_{k=1}^{n} \prod_{i=1}^{M} \pi \big ( y_{k}^{(i)} \big ).
\end{split}
\]
A universal test $\delta: \mathcal{Y}^{Mn} \rightarrow \{0, 1, \ldots, M\}$ will now also accommodate for an additional decision for the null hypothesis.  Correspondingly, the maximal error probability is now computed with the inclusion of the null hypothesis according to
\begin{align} 
e \big (\delta, (\mu, \pi) \big) \ \triangleq \  \max_{i = 0, 1, \ldots , M} \  
\sum_{y^{Mn}:\ \delta(y^{Mn}) ~\neq~ i}
p_{i} \big ( y^{Mn} \big ). \nonumber
\end{align}

With just one additional null hypothesis, contrary to the previous models with one outlier, it becomes impossible to achieve universally exponential consistency {\em even with the knowledge of the typical distribution.}  

\setcounter{prp}{3}
\begin{proposition} \label{prop-4}
For the setting with the additional null hypothesis,  there cannot exist a universally exponentially consistent test even when the typical distribution is known.
\end{proposition}

In typical applications such as environment monitoring and fraud detection, the consequence of a missed detection of the outlier can be much more catastrophic than that of a false positive.  In addition, Proposition \ref{prop-4} tells us that there cannot exist a universal test that yields exponential decays for both the conditional probability of false positive (under the null hypothesis) and the conditional probabilities of missed detection (under all non-null hypotheses).  Consequently, it is natural to look for a universal test fulfilling a lesser objective: attaining universally exponential consistency for conditional error probabilities under {\em all the non-null} hypotheses, while  seeking {\em only} universal consistency for the conditional error probability under the null hypothesis.  We now show that such a test can be obtained by slightly modifying \tmblue{the GL test in (\ref{eqn-univ-detector-2})}.  Furthermore, in addition to achieving universal consistency under the null hypothesis, this new test achieves the same exponent as in (\ref{eqn-err-exp}), 
(\ref{eqn-constraint}) in  Theorem \ref{thm-2} for the conditional error probabilities under {\em all} non-null hypotheses.

In particular, we modify the previous test in (\ref{eqn-univ-detector-2}) to allow for the possibility of deciding for the null hypothesis as follows.
\begin{align}
\hspace{-0.08in} \delta(y^{Mn}) \hspace{-0.02in} = \hspace{-0.02in} \left \{ 
\begin{array}{ll}  
\hspace{-0.08in} \mathop{\arg\max}\limits_{i=1, \ldots, M}\, \hat{p}_{i}^{\mbox{\scriptsize{univ}}}(y^{Mn}), 
\hspace{-0.08in} &
\mbox{if} \, \max\limits_{j \neq k} \frac{1}{n}\big (\hspace{-0.02in} \log \hat{p}_{j}^{\mbox{\scriptsize{univ}}}(y^{Mn})\\ & \hspace{0in}  - \log\hat{p}_{k}^{\mbox{\scriptsize{univ}}}(y^{Mn}) \big ) > \lambda_{n}, \\
\hspace{-0.08in} 0, 
& \mbox{otherwise,} 
\end{array} \right. 
\label{eqn-univ-detector-null}
\end{align}
where $\lambda_{n} = \Theta (\frac{\log n}{n})$ and the ties in the first case of (\ref{eqn-univ-detector-null}) are broken arbitrarily.

\setcounter{thm}{4}
\begin{theorem} \label{thm-4}
For every pair of distributions $\mu, \pi,\ \mu \neq \pi$, \tmblue{the} test in (\ref{eqn-univ-detector-null}) yields a positive exponent for the conditional probability of error under every non-null hypothesis $i = 1, \ldots, M$, and a vanishing conditional probability of error under the null hypothesis.  In particular, the achievable error exponent under every non-null hypothesis is the same as that given in (\ref{eqn-err-exp}), (\ref{eqn-constraint}), i.e., for each $i=1, \ldots, M,$ \tmblue{the} test in (\ref{eqn-univ-detector-null}) achieves
\begin{align} 
\hspace{-0.1in} & \lim_{n \rightarrow \infty} -\frac{1}{n} \log \left(\mathbb{P}_{i} \left \{ \delta \neq i\right \}\right) \nonumber \\
& \hspace{0.1in} =   
\min\limits_{q_{1}, \ldots, q_{{\scriptsize M}}} 
D \left( q_{1} \| \mu \right) \hspace{-0.02in} + \hspace{-0.02in} D \left( q_{2} \| \pi \right) \hspace{-0.02in} + \hspace{-0.02in} \ldots \hspace{-0.02in} + \hspace{-0.02in} D\left ( q_{{\scriptsize M}} \| \pi \right),  \label{eqn-thm4-exp-nonnull}
\end{align}
where the minimum above is over the set of $\left( q_1, \ldots, q_{\scriptsize M} \right)$ satisfying (\ref{eqn-constraint}).  In addition, the test also yields that
\begin{align}
\lim_{n \rightarrow \infty} \mathbb{P}_{0} \left \{ \delta \neq 0\right \} \  =\ 0.
\label{eqn-thm4-consistency-null}
\end{align}
\end{theorem}

Since under every non-null hypothesis, \tmblue{the} modified test in (\ref{eqn-univ-detector-null}) achieves the same exponent (the value of the optimization problem in (\ref{eqn-err-exp}), (\ref{eqn-constraint})) for the conditional error probability as \tmblue{the GL test in (\ref{eqn-univ-detector-2})} when the null hypothesis is absent, we get the following corollary by just observing that Theorem \ref{thm-3} was proved by finding a suitable lower bound for the value of the optimization problem in (\ref{eqn-err-exp}), (\ref{eqn-constraint}).

\setcounter{crl}{5}
\begin{corollary} \label{crl-5}
\tdbblue{For each $M \geq 3$} and under every non-null hypothesis $i = 1, \ldots, M,$ the exponent for the conditional error probability achievable by \tmblue{the} test in (\ref{eqn-univ-detector-null}) is lower bounded as
\tmred{
\begin{align}
\lim_{n \rightarrow \infty} -\frac{1}{n} \log \left(\mathbb{P}_{i} \left \{ \delta \neq i\right \}\right) \geq\mathop{\min\limits_{q \, \in \, \mathcal{P}(\cY)}}  2\, B (\mu \, , \, q )\, , \label{eqn-lowerbd'}
\end{align}
where the minimum above is over the set of $q$ such that
\begin{align}
D(q \| \pi) \, \leq \, \frac{1}{M-1} \big (2B(\mu, \pi) + C_{\pi} \big ), \nonumber
\end{align}
and $C_{\pi} \triangleq -\log \Big ( \min\limits_{y \in \cY} \, \pi (y) \Big ) < \infty.$ }
Consequently, as $M \rightarrow \infty$, this lower bound converges to the optimal error exponent $2B ( \mu, \pi ),$ i.e., for every $i = 1, \ldots, M,$ \tmblue{the} test in (\ref{eqn-univ-detector-null}) achieves
\begin{equation}
\lim_{M \rightarrow \infty}~  
\lim_{n \rightarrow \infty} -\frac{1}{n}~ \log \left(\mathbb{P}_{i} \left \{ \delta \neq i\right \}\right)
\ = \  2 B( \mu, \pi),   \nonumber
\end{equation}
while also yielding that
\begin{align}
\lim_{n \rightarrow \infty} \mathbb{P}_{0} \left \{ \delta \neq 0\right \} \  =\ 0.
\nonumber
\end{align}
\end{corollary}

\section{Universal Multiple Outlier Hypothesis Testing} \label{sec-multiple}

We now generalize our results in the previous section to models with multiple \tmblue{outlier sequences}.  With more than one \tmblue{outlier sequence}, it may be more natural to consider models for which the different \tmblue{outlier sequences} are distinctly distributed, and therefore our models will allow for this possibility.  Also, we shall consider two settings: the setting with a known number, \tmblue{say $T \geq 1$}, of outliers, and the setting with an unknown number of outliers, say up to $T$ outliers. \tmblue{It should be noted that the current model with $T = 1$ corresponds to the single outlier setting where the outlier distribution can vary according to the index of the outlier sequence.}

We start by describing a generic model with possibly distinctly distributed outliers, the number of which is not known.  As in the previous section, we denote the $k$-th observation of the $i$-th \tmblue{sequence} by $Y_k^{(i)} \in \cY, i = 1, \ldots, M,\ k = 1, \ldots, n.$  Most of the \tmblue{sequences} are commonly distributed according to the ``typical'' distribution $\pi \in \cP (\cY)$ except for a small (possibly empty) subset $S \subset \left \{1, \ldots, M \right \}$ of ``outlier'' \tmblue{sequences}, each of which is assumed to be distributed according to an outlier distribution $\mu_i,\ i \in S$.  {\em Nothing is known about $\left \{ \mu_i \right \}_{i=1}^{M}$ and $\pi$ except that each $\mu_i \neq \pi,\ i = 1, \ldots, M,$ and that all $\mu_i,\ i =1, \ldots, M,$ and $\pi$ have full supports.} In the following presentation, we sometimes consider the special case when all the \tmblue{outlier sequences} are identically distributed, i.e., $\mu_{i} = \mu$, $i= 1, \ldots, M$.

For the hypothesis corresponding to an outlier subset $S \subset \left \{1, \ldots, M \right \},\ \vert S \vert < \frac{M}{2}$, the joint distribution of all the observations is given by
\begin{align}
p_S \left( y^{Mn} \right)
&=\  p_S \left( \by^{(1)}, \ldots, \by^{(M)} \right)	\nonumber	\\
&=\  \prod_{k=1}^n	
	   \left \{ 
	   \prod_{i \in S}	\mu_i \left( y_k^{(i)} \right) 
	   \prod_{j \notin S}	\pi \left( y_k^{(j)} \right) 
	   \right \},		
\label{eqn-jointdis1}
\end{align}
where
\begin{align}
	\by^{(i)} = \left( y_1^{(i)}, \ldots, y_n^{(i)} \right),\ i = 1, \ldots, M.
	\nonumber
\end{align}
We refer to the unique hypothesis corresponding to the case with {\em no} outlier, i.e., $S = \emptyset,$ as the {\em null} hypothesis.  
In the following subsections, we shall consider different settings, each being described by a suitable set $\cal{S}$ comprising all possible outlier subsets.

The test for the 
outlier subset is done based on a {\em universal} rule $\delta:~ \cY^{Mn} \rightarrow \cal{S}$.  In particular, the test $\delta$ is not allowed to \tmblue{be a function of} $\left( \left \{ \mu_i \right \}_{i=1}^M, \pi \right)$.

For a universal test, the maximal error probability, which will be a function of the test and $\left( \left \{ \mu_i \right \}_{i=1}^M, \pi \right)$, is
\begin{align} 
e \left (\delta, \left( \left \{ \mu_i \right \}_{i=1}^M, \pi \right) \right) 
 \triangleq \ \max_{S \in \cal{S}} 
\hspace{-0.02in} \sum_{y^{Mn}:\ \delta \left( y^{Mn} \right) ~\neq~ S}
\hspace{-0.05in} p_{S} \big ( y^{Mn} \big ), \label{eqn-maxerror-mult}
\end{align}
and the corresponding error exponent is defined as
\begin{align}
\hspace{-0.03in} \alpha \hspace{-0.02in} \left ( \delta, \left( \left \{ \mu_i \right \}_{i=1}^M,\pi \right) \right ) \triangleq \hspace{-0.03in} \lim_{ n \rightarrow \infty}  
 \hspace{-0.04in} -\frac{1}{n} \log \, e \hspace{-0.02in} \left ( \delta, \left( \left \{ \mu_i \right \}_{i=1}^M, \pi \right) \right ). \nonumber
\end{align} 
A universal test $\delta$ is termed {\em universally exponentially consistent} if for every $\mu_i,\ i = 1, \ldots, M,\ \mu_i \neq \pi,$ it holds that 
\begin{align}
\alpha \left( \delta, 
\left( \left \{ \mu_i \right \}_{i=1}^M, \pi \right) \right) > 0. \nonumber
\end{align}  

\subsection{Models with Known Number of Outliers}
\label{sec-multoutlier-knownnumber}

We start by considering the case in which the number of outliers, denoted by \tmblue{$T \geq 1$}, is known at the outset, i.e., $\vert S \vert = T,$ for every $S \in \cal{S}$.  Unlike the model in Section \ref{sec-single} where the outlier sequence is always distributed according to a fixed distribution $\mu \neq \pi$ regardless of its identity $i = 1, \ldots, M,$  in our model for this subsection, the distributions of different outlier sequences $\mu_i,\ i \in S,$ {\em can} vary across \tmblue{the indices of the sequences}, $i \in S.$

\subsubsection{\tmblue{Generalized Likelihood Test}}
\label{sec-multoutlier-knownnumber-test}
We now give a summary of the GL test for the current models with multiple outlier sequences for both the setting when only $\pi$ is known and for the completely universal setting.

\tmblue{Conditioned on the outlier subset being $S \in \cal{S}$, the likelihood of $y^{Mn}$ is a function of the outlier indices, and the typical and outlier distributions (cf. (\ref{eqn-jointdis1})), i.e.,}
 \begin{align}
\tmblue{p_{S} \left( y^{Mn} \right)  \ = \ L \left (y^{Mn}, \{\mu_{i}\}_{i \in S}, \, \pi \right ).} \label{eqn-lfunction1}
\end{align}

\tmblue{When only $\pi$ is known, we compute the generalized likelihood of $y^{Mn}$ by replacing $\mu_{i}$ in (\ref{eqn-lfunction1}) with its ML estimate $\hat{\mu}_i \triangleq \gamma_i,$  $i \in S$, as}
\begin{align}
\tmblue{\hat{p}_{S}^{\mbox{\scriptsize{typ}}} \left( y^{Mn} \right)  \ = \ L \left (y^{Mn}, \{\hat{\mu}_{i}\}_{i \in S}, \, \pi \right ).} \label{eqn-univ-mult-pi-known}
\end{align}

\tmblue{Similarly, for the completely universal setting, we compute the generalized likelihood of $y^{Mn}$ by replacing the $\mu_{i}$ and $\pi$ in (\ref{eqn-lfunction1}) with their ML estimates $\hat{\mu}_i \triangleq \gamma_i$, $i \in S$, and $\hat{\pi}_{S} \triangleq \textstyle{\frac{\sum_{k \notin S} \gamma_{k}}{M-T}}$, as}
\begin{align}
\tmblue{\hat{p}^{\mbox{\scriptsize{univ}}}_{S} \left( y^{Mn} \right)  \ = \ L \left (y^{Mn}, \{\hat{\mu}_{i}\}_{i \in S}, \, \hat{\pi}_{S} \right ).} \label{eqn-univ-mult}
\end{align}

\tmblue{The test then selects the hypothesis under which the generalized likelihood is maximized (ties are broken arbitrarily), i.e.,}
\begin{align}
\tmblue{ \delta \big (y^{Mn} \big ) \ = \ 
\mathop{\mbox{argmax}}\limits_{S \subset \left \{  1, \ldots, M \right \},\ \vert S \vert = T} 
\ \hat{p}_{S}^{\mbox{\scriptsize{typ}}} }  \label{eqn-univ-mult-detector-pi}		
\end{align}
\tmblue{for the setting when only $\pi$ is known, and}
\begin{align} 
\tmblue{\delta \big (y^{Mn} \big ) \ = \ 
\mathop{\mbox{argmax}}\limits_{S \subset \left \{  1, \ldots, M \right \},\ \vert S \vert = T} 
\ \hat{p}_{S}^{\mbox{\scriptsize{univ}}} } \label{eqn-univ-mult-detector}		
\end{align}
\tmblue{for the completely universal setting, respectively.} \tmblue{It is straightforward to show using (\ref{eqn-probability}) that when only $\pi$ is known, the GL test in (\ref{eqn-univ-mult-detector-pi}) is equivalent to}
\begin{align}
\tmblue{ \delta \big (y^{Mn} \big ) \ = \ \mathop{\mbox{argmin}}\limits_{S \subset \left \{  1, \ldots, M \right \},\ \vert S \vert = T}  \ \sum\limits_{j \notin S} D(\gamma_{j} \| \pi)}, \label{eqn-univ-mult-detector-pi'}
\end{align}
\tmblue{and when neither $\pi$ nor $\left \{ \mu_i \right \}_{i=1}^M$ is known, the test in (\ref{eqn-univ-mult-detector}) is equivalent to}
\begin{align}
\tmblue{\delta \big (y^{Mn} \big ) \ = \ \mathop{\mbox{argmin}}\limits_{S \subset \left \{  1, \ldots, M \right \},\ \vert S \vert = T}  \ \sum\limits_{j \notin S}
D \left( \gamma_{j} \,\big \| \, \textstyle{\frac{\sum_{k \notin S} \gamma_{k}}{M-T}} \right). } \label{eqn-univ-mult-detector'}
\end{align}

\subsubsection{Results} 
\label{sec-multoutlier-knownnumber-results}

\setcounter{prp}{6}
\tdbblue{
\begin{proposition} \label{prop-6}  
For every fixed number of outliers \tmblue{$T \geq 1$}, when all the $\mu_i,\ i = 1, \dots, M,$ and $\pi$ are known, the optimal error exponent is equal to
\begin{align}
\min_{1 \leq i < j \leq M} 
 \ C \left(
 \mu_i \left( y \right) \pi \left( y'\right), 
 \pi \left( y \right) \mu_j \left( y' \right)
 \right).   \label{eqn-knownT-completely-nonuniv-exp}
\end{align}
\indent
When all outlier \tmblue{sequences} are identically distributed, i.e., $\mu_i = \mu \neq \pi,\ i = 1, \ldots, M,$ this optimal error exponent is independent of $M$ and is equal to 
\begin{align}
2 B  \left( \mu, \pi \right). 
\label{eqn-knownT-completely-nonuniv-exp-samemu}
\end{align}
\end{proposition}
}

\setcounter{thm}{7}
\tdbblue{
\begin{theorem} \label{thm-7}
For every fixed number of outliers \tmblue{$T \geq 1$}, when only $\pi$ is known but none of $\mu_i,\ i = 1, \dots, M$ is known, the error exponent achievable by \tmblue{the GL test in (\ref{eqn-univ-mult-pi-known})} is equal to
\begin{align}
\min_{1 \leq i \leq M}\  
 2B \left( \mu_i,  \pi 
 \right).   \label{eqn-knownT-semi-univ-exp}
\end{align}
\indent
When all \tmblue{outlier sequences} are identically distributed, i.e., $\mu_i = \mu, i = 1, \ldots, M,$ this achievable error exponent is equal to
\begin{align}
2 B  \left( \mu, \pi \right),
\end{align}
which, from Proposition \ref{prop-6}, is the optimal error exponent when $\mu$ is also known.
\end{theorem}
}

\begin{remark}{\em 
Since the tester in Proposition \ref{prop-6} is more capable (with the typical and outlier distributions known) than that in Theorem \ref{thm-7}, the optimal error exponent in (\ref{eqn-knownT-completely-nonuniv-exp}) must be no smaller than that in (\ref{eqn-knownT-semi-univ-exp}).  This is verified simply by noting that for every $i, j,\ 1 \leq i < j \leq M,$ we get from (\ref{def-Cher}) that}
\begin{align}
&C \left(
 \mu_i \left( y \right) \pi \left( y'\right), 
 \pi \left( y \right) \mu_j \left( y' \right)
 \right) \nonumber \\
& = \max_{s \in [0, 1]} 
 -\log \bigg [
\ \sum_{y, y' \in \cY \times \cY} 
  \big( \mu_i \left( y \right) \pi \left( y' \right) \big)^{s} \nonumber \\
 &\hspace{1.5in} \cdot \big( \pi \left( y \right) \mu_j \left( y' \right)  \big)^{1-s}
\bigg ] \nonumber \\
&  \geq
B \left( \mu_i, \pi \right) + B \left( \mu_j, \pi \right ) \nonumber\\
& \geq
\min \left( 
2B \left( \mu_i, \pi \right), 2B \left( \mu_j, \pi \right)
\right).
\end{align}
\nonumber
\end{remark}

As in Section \ref{sec-single}, for the current models, a test $\delta$ is {\em universally exponentially consistent} if for every $\mu_i,\ i = 1, \ldots, M,\ \mu_i \neq \pi,$ it holds that 
$\alpha \left( \delta, \left( \left \{ \mu_i \right \}_{i=1}^M , \pi \right) \right) > 0.$  

\begin{theorem} \label{thm-8}
For every fixed number of outliers $1 \leq T < \frac{M}{2}$, the GL test $\delta$ in (\ref{eqn-univ-mult-detector}) is universally exponentially consistent. Furthermore, for every 
$\left \{ \mu_{i} \right \}_{i=1}^{M}, \pi 
,\ \mu_i \neq \pi,\ i = 1, \ldots, M$, it holds that
\begin{align} 
&\alpha \left ( \delta, \left( \left \{ \mu \right \}_{i=1}^{M}, \pi \right) \right ) \nonumber \\
& \hspace{-0.04in} = \hspace{-0.05in} \mathop{\min_{S, S' \subset \left \{1, \ldots, M \right \}}}_{\vert S \vert = \vert S' \vert =T}
\, \min\limits_{q_{1}, \ldots, q_{{\scriptsize M}}} 
\hspace{-0.02in} \Big ( \sum_{i \in S} D \left( q_i \| \mu_i \right) + \hspace{-0.02in}
\sum_{j \notin S} D\left ( q_j \| \pi \right) \Big ),  \label{eqn-thm6-err-exp}
\end{align}
where the inner minimum above is over the set of $\left( q_1, \ldots, q_{\scriptsize M} \right)$ such that
\begin{align} 
\sum_{i \notin S} 
D \left ( q_{i} \, \Big \| \, \textstyle{ \frac{\sum_{k \notin S} q_{k}}{M-T}  } \right ) 
\geq 
\sum_{i \notin S'} 
D \left ( q_i \, \Big \| \, \textstyle{\frac{\sum_{k \notin S'} q_{k}}{M-T} } \right ).  
\label{eqn-thm6-constraint}
\end{align}
\end{theorem}
Note that universally exponential consistency does not imply that
\begin{align}
	\lim\limits_{M \rightarrow \infty}
	~\alpha \left( \delta, \left( \left \{ \mu_i \right \}_{i=1}^M , \pi \right) \right)
	\ >\ 0.
\label{def-univ-asymp-exp-consistent}
\end{align}
\tmblue{Furthermore, it follows from Proposition \ref{prop-6} that (\ref{def-univ-asymp-exp-consistent}) is not possible if $\left( \left \{ \mu_i \right \}_{i=1}^M , \pi \right)$ satisfies that} 
\begin{align}
	\lim\limits_{M \rightarrow \infty}\ 
	\min\limits_{1 \leq i < j \leq M}\ 
	C \left(
		\mu_i \left( y \right) \pi \left( y' \right),
	     \pi \left( y \right) \mu_j \left( y' \right)
	\right)
	\ =\ 0.
\label{def-univ-asymp-exp-consistent-impossible}
\end{align}
Consequently, a natural question that arises is whether there exists a test that
achieves a positive limiting error exponent as $M$ approaches infinity whenever the optimal error exponent  does not vanish with $M$, i.e., its achievable error exponent satisfies (\ref{def-univ-asymp-exp-consistent}) whenever (\ref{def-univ-asymp-exp-consistent-impossible}) {\em does not} hold. Such a test is said to be {\em asymptotically exponentially consistent.}

\begin{theorem} \label{thm-9}  
For every $M \geq 3$, and every fixed number of outliers \tmblue{$1 \leq T < \frac{M}{2}$}, the error exponent achievable by \tmblue{the GL test in  (\ref{eqn-univ-mult-detector})} is lower bounded by 
\tmred{
\begin{align}
\mathop{\min\limits_{q \, \in \, \mathcal{P}(\cY)}}\ \min_{i=1, \ldots, M}
\ 2\, B (\mu_i \, , \, q )\, , \label{eqn-thm9-lowerbd}
\end{align}
where the outer minimum above is over the set of $q$ such that
\begin{align}
\hspace{-0.02in} D(q \| \pi) \leq \hspace{-0.03in} \frac{1}{M-T} \Big ( \min\limits_{1\leq i < j \leq M} \hspace{-0.03in} C & \left( \mu_i (y) \pi(y'), \pi(y)\mu_{j}(y') \right)\nonumber \\  
&\hspace{0.7in} + \ T C_{\pi}  \Big ),  \nonumber
\end{align}
and $C_{\pi} \triangleq -\log \Big ( \min\limits_{y \in \cY} \, \pi (y) \Big ) < \infty.$}

Furthermore, as $M \rightarrow \infty$, the error exponent achievable by the test in (\ref{eqn-univ-mult-detector}) converges as
\begin{equation}
\lim_{M \rightarrow \infty} 
\alpha \left ( 
\delta, \left ( \left \{ \mu_i \right \}_{i=1}^M, \pi \right) 
\right ) =  
\lim_{M \rightarrow \infty} \ \min\limits_{i=1, \ldots, M} 
2 B( \mu_{i}, \pi),   \label{eqn-thm9-limit}
\end{equation}
which from (\ref{eqn-knownT-semi-univ-exp}) of Theorem \ref{thm-7} is also the limit of the achievable error exponent of the test in \tmblue{(\ref{eqn-univ-mult-detector-pi})} using the knowledge of the typical distribution. The limiting error exponent on the right-side of (\ref{eqn-thm9-limit}) is always positive whenever (\ref{def-univ-asymp-exp-consistent-impossible}) does not hold.

When all outlier sequences are identically distributed, i.e., $\mu_i = \mu \neq \pi,\ i = 1, \ldots, M,$ \tmblue{the test in  (\ref{eqn-univ-mult-detector})} 
achieves the optimal error exponent asymptotically as the number of sequences approaches infinity. i.e.,
\begin{align}
\lim_{M \rightarrow \infty} 
\alpha \left ( 
\delta, \left ( \mu, \pi \right) 
\right ) = 
2 B \left( \mu, \pi \right). \label{eqn-thm9-eff-iden-outliers}
\end{align}
\end{theorem}

\subsection{Models with Unknown Number of Outliers}

\tmblue{In this section, we look at the setting where there is uncertainty in the number of outliers, i.e.,} not all hypotheses in $\cal{S}$ have the same number of outliers.
\tmblue{It is also assumed that for a fixed number of outliers $k = 0, 1, 2, \ldots$, $\cal{S}$ either contains {\em all} hypotheses with $k$ outliers, or {\em none} of them} 

\subsubsection{Models with Identical Outliers}

\tmblue{We now show that when all outlier sequences are identically distributed, even without knowing the number of outliers exactly (assumed in Section \ref{sec-multoutlier-knownnumber}), the GL test is universally exponentially consistent as long as we know that there are always some outliers. }

\tmblue{Now {\em with the assumption of identically distributed outliers being taken strictly}, we compute the generalized likelihood of $y^{Mn}$ by replacing the $\mu_{i}, i \in S,$ and $\pi$ in \tmblue{(\ref{eqn-lfunction1})} with their ML estimates} $\hat{\mu}_S = \hat{\mu}_i \triangleq \textstyle{\frac{\sum_{k \in S} \gamma_k}{\vert S \vert}}$, and $\hat{\pi}_S \triangleq \textstyle{\frac{\sum_{k \notin S} \gamma_{k}}{M - \vert S \vert}}$, as
\begin{align}
\hat{p}^{\mbox{\scriptsize{univ}}}_{S} \left( y^{Mn} \right)  \ = \ L \left (y^{Mn}, 
\hat{\mu}_S, \, \hat{\pi}_{S} \right ). \label{eqn-univtestsecB2}
\end{align}

\tmblue{The test then selects the hypothesis under which the generalized likelihood in (\ref{eqn-univtestsecB2}) is maximized (ties are broken arbitrarily), i.e.,}
\begin{align}
\tmblue{\delta(y^{Mn}) \ = \ \mathop{\mbox{argmax}}\limits_{\scriptsize{S} \in \cal{S}} \ \hat{p}^{\mbox{\scriptsize{univ}}}_{S} \left( y^{Mn} \right)}. \label{eqn-univtestsecB1}
\end{align}
It is straightforward to show using (\ref{eqn-probability}) that the GL test in (\ref{eqn-univtestsecB1}) is equivalent to
\begin{align}
\delta \big (y^{Mn} \big ) \ =\ \mathop{\mbox{argmin}}\limits_{S \in \cal{S}} \sum_{i \in S} & D \big (\gamma_{i} \big  \| \textstyle \frac{\sum_{k \in S} \gamma_{k}}{T} \big ) \nonumber \\
& + \displaystyle \sum_{j \notin S}D \big (\gamma_{j} \big \| \textstyle \frac{\sum_{k \notin S} \gamma_{k}}{M-T} \big ). \label{eqn-univtestsecB1'}
\end{align}


\begin{theorem} \label{thm-10}  
\tmblue{When there are at most $T$, $1 \leq T < M/2$, number of outliers in each hypothesis, and all the outlier sequences are identically distributed, the GL test in (\ref{eqn-univtestsecB1})} is universally exponentially consistent for every hypothesis set excluding the null hypothesis.  On the other hand, when the hypothesis set contains the null hypothesis,  there cannot exist a universally exponentially consistent test even when the typical distribution is known.
\end{theorem}

\begin{remark} {\em When the null hypothesis is present, we can make a suitable modification to \tmblue{the test in (\ref{eqn-univtestsecB1})} similar to (\ref{eqn-univ-detector-null}) to get a universal test that achieves a positive exponent for every conditional error probability, conditioned on any non-null hypothesis, and additional consistency under the null hypothesis.}
\end{remark}

\subsubsection{\tdblue{Models with Distinct Outliers}}

\tdblue{Our last pessimistic result shows that when the \tmblue{outlier sequences} can be distinctly distributed in an arbitrary manner, the assumption of a known number of outliers adopted in Section 
\ref{sec-multoutlier-knownnumber} is indeed critical, as a lack thereof would make it impossible to construct a universally exponentially consistent test even when there are {\em always some} outliers.}

\tdblue{\begin{theorem} \label{thm-11} 
When the \tmblue{outlier sequences} can be distinctly distributed, there cannot exist a universally exponentially consistent test, even when the typical distribution is known and when the null hypothesis is excluded, i.e., \tmblue{there is at least one outlier regardless of the hypothesis}. 
\end{theorem}}


\begin{remark}
{\em The negative result in Theorem \ref{thm-11} should not be taken with extreme pessimism.  It should be viewed as a theoretical result that holds {\em only when} each of the outliers can be arbitrarily distributed.  In practice, there will likely be modeling constraints that would confine the set of all possible tuples of the distributions of all outliers.  An extreme case of such constraints is when all the outliers are forced to be identically distributed, which is when universally exponential consistency is indeed attained (cf. Theorem \ref{thm-10}) if the null hypothesis is excluded.  An interesting future research direction would be to characterize the ``least'' stringent joint constraint on the distributions of the outliers that still allows us to construct universally exponentially consistent tests.}
\end{remark}

\section{Discussion}
In this paper, we formulated and studied the problem of outlier hypothesis testing in a completely universal setting. \tmblue{Our main contribution was in proving that GL tests yield exponentially decaying probability of error under various settings.}
\tmblue{In particular, for the case with exactly one outlier, the GL test was shown to be universally exponentially consistent.} We also provided a characterization of the error exponent achievable by \tmblue{the GL test} for each $M \geq 3$. Surprisingly \tmblue{the GL test} is not only universally exponentially consistent, but also asymptotically optimal as the number of sequences goes to infinity. Specifically, as $M$ goes to infinity, the error exponent achievable by \tmblue{the GL test} converges to the absolutely optimal error exponent when both the outlier and typical distributions are known. When there is an additional null hypothesis, a suitable modification of \tmblue{the GL test} was shown to achieve exponential consistency under each hypothesis with the outlier, and consistency under the null hypothesis universally.  Under every non-null hypothesis, this modified test achieves the same error exponent as that achievable when the null hypothesis is excluded. We then extended our models to cover the case with \tmblue{more than one outlier}. For models with a known number of outliers, the distributions of the outliers could be distinct as long as each of them differs from the typical distribution.  \tmblue{The GL test} was shown to be universally exponentially consistent. Furthermore, we characterized the limiting error exponent achieved by \tmblue{such a test}, and established its universally asymptotically exponential consistency. When the number of outliers is not known, it was shown that the assumption of the outliers being identically distributed and the exclusion of the null hypothesis were both essential for existence of universally exponentially consistent test. In particular, for models with an unknown number of {\em identically} distributed outliers, \tmblue{the GL test} is universally exponentially consistent when the null hypothesis is excluded. When the null hypothesis is included, \tmblue{a slight modification of the GL test} was shown to achieve a positive error exponent under every non-null hypothesis, and also consistency under the null hypothesis universally.  For models with an unknown number of {\em distinctly} distributed outliers, it was shown that even when the typical distribution is known and when the null hypothesis is excluded, a universally exponentially consistent test cannot exist. 

\tmblue{We end with a discussion of possible extensions of our results. First, it is worth noting that the results in our paper only apply to the case where the observation alphabet is finite. A useful extension would be to generalize our results to the case with more general observation alphabets \cite{zou-lian-poor-shi-ieeetit-2014}. Another interesting extension would be to consider models with the size of the alphabet being large compared to the number of samples from each sequence.  Such a situation is usually formulated as one in which the alphabet is allowed to grow with the number of 
samples \cite{huan-meyn-2012, kell-wagn-tula-visw-2013}. For the universal outlier hypothesis testing problem, a natural question that arises is how fast can the alphabet size be allowed to grow while still retaining universal consistency or exponential consistency.  Finally, to bridge the theory with practice, it remains to investigate the extent to which the GL tests are applicable in the applications mentioned in the introduction such as severe weather prediction, environment monitoring in sensor networks, network intrusion, voting irregularity analysis, spectrum sensing, and high frequency trading.}

\appendices \label{sec-app}


\section{} \label{sec-app1}
\vspace{-0.0in}
Our proofs rely on the following lemmas.


\begin{lemma}\label{lm-1}
Let $\bY^{(1)}, \ldots, \bY^{(J)}$ be mutually independent random vectors with each $\bY^{(j)}$, $j=1, \ldots, J$, being $n$ i.i.d. repetitions of a random variable distributed  according to $p_{j} \in \mathcal{P}\left( \cY \right)$. Let $A_{n}$ be the set of all $J$ tuples $\left( \by^{(1)}, \ldots, \by^{(J)} \right) \in \cY^{Jn}$ whose empirical distributions \tdbblue{$\left( \gamma_{1}, \ldots, \gamma_{J} \right) = \left( \gamma_{\by^{(1)}}, \ldots, \gamma_{\by^{(J)}} \right)$} lie in a closed set $E\in \mathcal{P} \left( \cY \right)^{J}.$ 
Then, it holds that
\vspace{-0.06in}
\begin{align}
& \lim_{n \rightarrow \infty} -\frac{1}{n} \log \mathbb{P} \Big \{ \left (\bY^{(1)}, \ldots, \bY^{(J)} \right ) \in A_{n} \Big \} 
= \nonumber \\
& \hspace{1.1in} \min_{(q_{1}, \ldots, q_{J}) \, \in \, E} \ \sum_{j=1}^{J} \ D\left( q_{j} \| p_{j} \right).  \label{eqn-lm1-2}
\end{align}
\end{lemma}
\vspace{-0.02in}
\begin{proof}
Let $\overline{E}$ be the set of all joint distributions in $\mathcal{P} \left( \mathcal{Y}^J \right)$ with the tuple of their corresponding marginal distributions lying in $E$.  It now follows from the closeness of $E$ in $\mathcal{P}\left( \cY \right)^{J}$ and the compactness of $\mathcal{P} \left( \mathcal{Y}^J \right)$ that $\overline{E}$ is also closed in $\mathcal{P} \left( \mathcal{Y}^J \right)$.  Let $\overline{A}_{n}$ be the set of all $J$ tuples 
$\left( \by^{(1)}, \ldots, \by^{(J)} \right) =  
\Big( \big( y^{(1)}_1, \ldots, y^{(1)}_n \big), \ldots, \big( y^{(J)}_1, \ldots, y^{(J)}_n \big) \Big) 
\in \cY^{Jn}$ whose joint empirical distribution lies in a closed set $\overline{E} \in \mathcal{P}\left( \cY^J \right).$  The lemma then follows by observing that 
$\mathbb{P} \Big \{ \big (\bY^{(1)}, \ldots, \bY^{(J)} \big ) \in A_{n} \Big \} ~=~
\mathbb{P} \Big \{ 
\Big ( \big( y^{(1)}_1, \ldots, y^{(1)}_n \big), \ldots, \big( y^{(J)}_1, \ldots, y^{(J)}_n \big) \Big )
\in \overline{A}_{n} 
\Big \},$ and by invoking Sanov's theorem to compute the exponent of the latter probability, i.e., 
\begin{align}
&\lim_{n \rightarrow \infty} -\frac{1}{n} \log 
\mathbb{P} \Big \{ \left (\bY^{(1)}, \ldots, \bY^{(J)} \right ) \in A_{n} \Big \} \nonumber \\  
& = \lim_{n \rightarrow \infty} -\frac{1}{n} \log 
\mathbb{P} \Big \{ 
\Big ( \big ( y^{(1)}_1, \ldots, y^{(1)}_n \big ), \ldots, \nonumber \\
&\hspace{1.5in}  \big ( y^{(J)}_1, \ldots, y^{(J)}_n \big ) \Big )
\in \overline{A}_{n} 
\Big \}		
\nonumber \\
& = \ 
\min_{q \, \in \, \overline{E}} \  D \left( q \| p_1 \times ~\ldots~ \times p_J \right)
\nonumber  \\
& = \ 
\min_{(q_{1}, \ldots, q_{J}) \, \in \, E} \ \sum_{j=1}^{J} \ D\left( q_{j} \| p_{j} \right)
\nonumber 
\end{align}
\end{proof}

\begin{lemma} \label{lm-2}
For any two pmfs $p_{1}$, $p_{2} \in \mathcal{P}(\cY)$ with full supports, it holds that
\tdbblue{
\begin{align}
2 B \left( p_{1} \, , \, p_{2} \right) = \min \limits_{q \, \in \, \mathcal{P}(\cY)} \Big (D \left( q \| p_1 \right) 
+ D \left( q \| p_2 \right) \Big ). \label{eqn-BDandD}
\end{align}}
In particular, the minimum on the right side of (\ref{eqn-BDandD}) is achieved by
\begin{align}
q^{\star} \ = \ \frac{p_1^{\frac{1}{2}}(y) p_2^{\frac{1}{2}}(y)}{\sum\limits_{y \in \cY} p_1^{\frac{1}{2}}(y) p_2^{\frac{1}{2}}(y)}, \ y \in \cY. \label{eqn-BDandD-minval}
\end{align}
\end{lemma}
\begin{proof}
It follows from the concavity of the logarithm function that
\begin{align}
D \left( q \| p_{1} \right) + 
D \left( q \| p_{2} \right) 
& =  \sum_{y \in \cY} 
q(y) \log \frac{q^{2}(y)}{p_{1}(y)p_{2}(y)} \nonumber \\
& = - 2 \sum_{y \in \cY} 
q(y) \log \frac{p^{\frac{1}{2}}_{1}(y) p^{\frac{1}{2}}_{2}(y)}{q(y)} \nonumber \\ 
& \geq -2 \log \bigg ( \sum_{y \in \cY} p^{\frac{1}{2}}_{1}(y) p^{\frac{1}{2}}_{2}(y) \bigg ) \label{eqn-BDandD1}  \\
&  = 2 B(p_{1}, p_{2}). \nonumber 
\end{align}
In particular, equality is achieved in (\ref{eqn-BDandD1}) by $q(y) = q^{\star}(y)$ in (\ref{eqn-BDandD-minval}).  

It is interesting to note that from (\ref{eqn-BDandD}), we recover the known inequality discovered in \cite{hoef-wolf-amstat-1958}:
\begin{align}
2 B \left( p_{1} \, , \, p_{2} \right) \ \leq \ 
\min \left ( D \left( p_2 \| p_1 \right), D \left( p_1 \| p_2 \right) \right),
\end{align}
by evaluating the argument distribution $q$ on the right-side of (\ref{eqn-BDandD}) by $p_1$ and $p_2,$ respectively.
\end{proof}


\begin{lemma} \label{lm-3}
For any two pmfs $p_{1}$, $p_{2} \in \mathcal{P}(\cY)$ with full supports, it holds that
\begin{align}
C \left( p_{1}, p_{2} \right) \ \leq\ 
2B \left( p_1, p_2 \right).
\nonumber
\end{align}
\end{lemma}
\begin{proof}  The proof follows from an alternative characterization (instead of (\ref{def-Cher})) of the $C \left( p_{1}, p_{2} \right)$ as (cf. \cite{Levy-2008})
\begin{align}
C \left( p_{1}, p_{2} \right) \ =\ 
\min\limits_{q \in \cP ( \cY )} 
\mbox{max} \left( D \left( q \| p_1 \right), D \left( q \| p_2 \right) \right).
\label{def-Cher-alt}
\end{align}
and upon noting that the objective function for the optimization problem in (\ref{def-Cher-alt}) is always no larger than that for the one in (\ref{eqn-BDandD}). 
\end{proof}


\section{}
\subsection{Proof of Theorem \ref{thm-1}}
%
\tdbblue{Since we consider the error exponent as $n$ goes to infinity while $M$ and hence the number of hypotheses is fixed, the ML test, which maximizes the error exponent for the average error probability (averaged over all hypotheses),  will also achieve the best error exponent for the maximal error probability.}
In particular, for any $y^{Mn} = \left ( \by^{(1)}, \ldots, \by^{(M)}  \right) \in \cY^{Mn}$, with $\gamma_{\by^{(i)}} = \gamma_{i}$, $i = 1, \ldots, M$, conditioned on the $i$-th \tmblue{sequence} being the outlier, \tmblue{applying the identity in (\ref{eqn-probability}),  it now follows from (\ref{eqn-jointdis})} that 
the ML test is
\begin{align}
\delta(y^{Mn}) \ = \ \mathop{\mbox{argmin}}_{i=1, \ldots, M} \ U_{i}(y^{Mn}),\nonumber
\end{align}
where for each $i=1, \ldots, M$,
\begin{align}
U_{i}(y^{Mn}) \ \triangleq \ D\left ( \gamma_{i} \| \mu \right ) + \sum_{j \neq i} D\left ( \gamma_{j} \| \pi \right ). \label{eqn-negloglikelihood1}
\end{align}

By the symmetry of the problem, it is clear that $\mathbb{P}_{i} \left \{ \delta \neq i \right \}$ is the same for every $i=1, \ldots, M$; hence,
\begin{equation} 
\max\limits_{i =1, \ldots, M}  \mathbb{P}_{i}  \left \{ \delta \neq i \right \} \ = \ \mathbb{P}_{1} \left \{ \delta \neq 1 \right \}. \nonumber
\end{equation}
It now follows from
\begin{align}
\tmblue{\mathbb{P}_{1} \left \{ \delta \neq 1 \right \} \ = \ \mathbb{P}_{1} \left( \cup_{j \neq 1} \{ U_{1} \geq U_{j} \} \right )},
\end{align}
that
\begin{align}
\mathbb{P}_{1} \left \{ U_{1} \geq U_{2} \right \} \ \leq \ \mathbb{P}_{1}  \left \{ \delta \neq 1 \right \} \ \leq \ \sum_{j=2}^{M} \mathbb{P}_{1} \left \{  U_{1} \geq U_{j}  \right \} .  \label{eqn-pairwise-error}
\end{align} 
Next, we get from (\ref{eqn-negloglikelihood1}) that
\begin{align}
\mathbb{P}_{1} \left \{ U_{1} \geq U_{2} \right \} \ = \ \mathbb{P}_{1}  \{ & D\left ( \gamma_{1} \| \mu \right ) + D\left ( \gamma_{2} \| \pi \right ) \nonumber \\
& \ \ \geq  D\left ( \gamma_{1} \| \pi \right ) +  D\left ( \gamma_{2} \| \mu \right ) \}.  \nonumber
\end{align}
Applying Lemma \ref{lm-1} with $J=2$, $p_1=\mu$, $p_2=\pi$, and 
\begin{align}
E = \Big \{  \left( q_{1}, q_{2} \right):~
& D\left ( q_{1} \| \mu \right ) + D\left ( q_{2} \| \pi \right )  \nonumber \\
& \ \ \ \geq D\left ( q_{1} \| \pi \right ) +  D\left ( q_{2} \| \mu \right ) 
\Big \}, \nonumber
\end{align}
we get that the exponent for $\mathbb{P}_{1} \left \{ U_{1} \geq U_{2} \right \}$ is given by the value of the following optimization problem
\tmred{
\begin{align}
\mathop{\min\limits_{q_{1}, q_{2} \, \in \, \mathcal{P}(\cY)}} \Big (D \left ( q_{1} \| \mu \right ) + D \left ( q_{2} \| \pi \right ) \Big ), \label{eqn-optimization1}
\end{align}
where the minimum above is over the set of $q_{1}, q_{2}$ such that
\begin{align}
D \left ( q_{1} \| \mu \right ) + D \left ( q_{2} \| \pi \right ) \,\geq \, D \left (q_{1} \| \pi  \right ) + D \left ( q_{2} \| \mu \right ). \nonumber 
\end{align}
}
\noindent
Note that the objective function in (\ref{eqn-optimization1}) is convex in $(q_{1}, q_{2})$, and the constraint is linear in $(q_{1}, q_{2})$. It then follows that the optimization problem in (\ref{eqn-optimization1}) is convex.  Consequently, strong duality holds for the optimization problem (\ref{eqn-optimization1}) \cite{boyd-vand-convopt-book-2004}.
Then by solving the Lagrangian dual of (\ref{eqn-optimization1}), its solution can be easily computed to be $2B( \mu, \pi)$.

By the symmetry of the problem, the exponents of $\mathbb{P}_{1} \left \{ U_{1} \geq U_{i} \right \}$, $i \neq 1$, are the same, i.e., for every $i = 2, \ldots, M,$ we get
\begin{align}
\lim_{n \rightarrow \infty} -\frac{1}{n} \log \mathbb{P}_{1}  \left \{ U_{1} \geq U_{i} \right \} \ = \ 2B(\mu, \pi). \label{eqn-errexp1}
\end{align}
From (\ref{eqn-pairwise-error}), (\ref{eqn-errexp1}), 
\tdbblue{using the union bound
and that $\lim_{n \rightarrow \infty} \frac{\log M}{n} = 0,$ we get that}
\begin{align}
\lim_{n \rightarrow \infty} -\frac{1}{n} \log \mathbb{P}_{1}  \left \{ \delta \neq 1 \right \} \ =\ 2B(\mu, \pi). \label{eqn-thm1-errexp}
\end{align}

It is now left to prove that when only $\pi$ is known, \tmblue{the GL test in (\ref{eqn-detector-pi-known-2}) and (\ref{eqn-detector-pi-known-2'})} also achieves the optimal error exponent $2B(\mu, \pi)$. 

\tmblue{For each $i=1, \ldots, M,$ denote the test statistic in (\ref{eqn-detector-pi-known-2'}) as
\vspace{-0.13in}
\begin{align}
 U_{i}^{\mbox{\scriptsize{typ}}} \triangleq D(\gamma_{i} \| \pi). \nonumber
\end{align}}
It follows from the same argument leading to (\ref{eqn-thm1-errexp}) that
\begin{align}
&\lim_{n \rightarrow \infty} - \frac{1}{n} \log \mathbb{P}_{1}\{ \delta' \neq 1\} \nonumber \\
&= \lim_{n \rightarrow \infty} - \frac{1}{n} \log  \mathbb{P}_{1} \Big \{ U_{1}^{\mbox{\scriptsize{typ}}} \ \leq \ U_{2}^{\mbox{\scriptsize{typ}}} \Big \}. 
\label{eqn-thm1-errexp1}
\end{align}
The exponent on the right-side of (\ref{eqn-thm1-errexp1}) can be computed by applying Lemma \ref{lm-1} with $J=2, \, p_{1} = \mu, \, p_{2} = \pi$, and
\begin{align}
E \ =\ \big \{ \left(q_{1}, q_{2} \right):~ D(q_{2} \| \pi) \geq D(q_{1} \| \pi ) \big \} \nonumber
\end{align}
to be
\vspace{-0.1in}
\begin{align}
\mathop{\min\limits_{q_{1}, q_{2} \in \mathcal{P}(\cY)}}_{
D \left ( q_{2} \| \pi \right ) \ \geq\  D \left (q_{1} \| \pi  \right ) }  \hspace{-0.0in} \Big ( D \left ( q_{1} \| \mu \right ) + D \left ( q_{2} \| \pi \right ) \Big )
\label{eqn-optimization2}
\end{align}
The optimal value of (\ref{eqn-optimization2}) can be computed as follows
\begin{align}
&
\mathop{\min\limits_{q_{1}, q_{2} \in \mathcal{P}(\cY)}}_{D \left ( q_{2} \| \pi \right ) \geq D \left (q_{1} \| \pi  \right ) } \Big ( D \left ( q_{1} \| \mu \right ) + D \left ( q_{2} \| \pi \right ) \Big ) 
\label{eqn-optimization22}  \\ 
&  
\geq\ \ \min_{q_{1}} \ \Big ( D \left ( q_{1} \| \mu \right ) + 
D \left ( q_{1} \| \pi \right ) \Big )
\label{eqn-inequlity2} \\
&  
= \ \ 2 B(\mu, \pi), \label{eqn-inequality3} 
\end{align}
where the inequality in (\ref{eqn-inequlity2}) stems from substituting the constraint in (\ref{eqn-optimization22}) into the objective function, and the equality in (\ref{eqn-inequality3}) follows from Lemma \ref{lm-2}. Since the minimum in (\ref{eqn-inequlity2}) is achieved by $q_{1}=q^{\star}$ in (\ref{eqn-BDandD-minval}) with $p_{1}= \mu$, $p_{2}= \pi$, and $q_{1}=q_{2}=q^{\star}$ satisfy the constraint in (\ref{eqn-optimization22}), the inequality in (\ref{eqn-inequlity2}) is in fact an equality.

\vspace{-0.09in}
\subsection{Proof of Theorem \ref{thm-2}}
\vspace{-0.04in}
For each $i=1, \ldots, M,$ denote the test statistic in (\ref{eqn-univ-detector-2'}) as
\vspace{-0.05in}
\begin{align}
U_{i}^{\mbox{\scriptsize{univ}}} \triangleq \sum_{j \neq i} D \big (\gamma_{j} \big \| \textstyle \frac{\sum_{k \neq i} \gamma_{k}}{M-1} \big). \label{eqn-defU-univ}
\end{align}
The same argument leading to (\ref{eqn-thm1-errexp}) yields that
\vspace{-0.08in}
\begin{align}
& \lim_{n \rightarrow \infty}  - \frac{1}{n} \log \mathbb{P}_{1}\{ \delta \neq 1\} \nonumber \\
& = \  \lim_{n \rightarrow \infty} - \frac{1}{n} \log  \mathbb{P}_{1} \Big \{ U_{1}^{\mbox{\scriptsize{univ}}} \ \geq \ 
U_{2}^{\mbox{\scriptsize{univ}}} \Big \}, 
\label{eqn-thm2-errexp1}
\end{align}

By applying Lemma \ref{lm-1} with $J=M,\ p_{1} = \mu,\ p_{j} = \pi,\ j=2, \ldots, M$, and  
\vspace{-0.04in}
\begin{align}
 E = \bigg \{ 
\left( q_{1}, \ldots, q_{M} \right) \,: \,
& \sum_{j \neq 1} \ D \Big ( q_{j} \, \Big \| \, \textstyle{\frac{\sum_{k \neq 1} q_{k}}{M-1} } \Big )
\nonumber \\
&\hspace{-0.1in} \geq \  
\displaystyle{\sum_{j \neq 2}} \ D \Big (q_{j} \, \Big \| \, 
\textstyle{ \frac{\sum_{k \neq 2} q_{k}}{M-1} } \Big ) 
\bigg \}, \label{eqn-pf-thm2-constraintset}
\end{align}
the exponent on the right-side of (\ref{eqn-thm2-errexp1}) can be computed to be
\begin{align}
\mathop{\min\limits_{\left ( q_{1}, \ldots, q_{M} \right )\, \in \, E}} D \left ( q_{1} \| \mu \right ) + D \left ( q_{2} \| \pi \right )+ \ldots +D\left ( q_{M} \| \pi \right ). \label{eqn-optimization3}
\end{align}
\vspace{-0.02in}
Unlike the convex optimization problems in (\ref{eqn-optimization1}) and (\ref{eqn-optimization2}), the optimization problem in (\ref{eqn-optimization3}) for the completely universal setting is much more complicated, and a closed-form solution is not available. However, we show that the value of (\ref{eqn-optimization3}) is strictly positive \tdbblue{for every $\mu \neq \pi$.} In particular, it is not hard to see that the objective function is continuous in $q_{1}, \ldots, q_{M}$ and the constraint set $E$ is compact. \tdbblue{Therefore the minimum in (\ref{eqn-optimization3}) is achieved by some $(q^{\star}_{1}, \ldots, q^{\star}_{M}) \in E$. Note that the objective function in (\ref{eqn-optimization3}) is always nonnegative. In order for the objective function in (\ref{eqn-optimization3}) to be zero, the minimizing $(q^{\star}_{1}, \ldots, q^{\star}_{M})$ has to satisfy that $q^{\star}_{1}=\mu$, $q^{\star}_{i}=\pi$, $i=2, \ldots, M$.  Since this collection of distributions is not in the constraint set $E$ in (\ref{eqn-pf-thm2-constraintset}), we get that the optimal value of (\ref{eqn-optimization3}) is strictly positive for every $\mu \neq \pi$.}

\vspace{-0.15in}
\subsection{Proof of Theorem \ref{thm-3}}
\vspace{-0.05in}
By the continuity of the objective function on the right-side of (\ref{eqn-err-exp}) and the compactness of the constraint set (\ref{eqn-constraint}), for each $M \geq 3$, the optimal value on the right-side of (\ref{eqn-err-exp}), denoted by $V^{\star}$, is achieved by some $\left( q^{\star}_{1}, \ldots, q^{\star}_{M} \right)$. 
It then follows from (\ref{eqn-err-exp}) and (\ref{eqn-constraint}) that
\vspace{-0.02in}
\begin{align}
V^{\star} & \geq D (q^{\star}_{1} \, \| \, \mu ) +  \sum_{j \neq 1}  D \left ( q^{\star}_{j} \, \| \, \pi \right ) \nonumber \\
& \hspace{0.2in} - \sum_{j \neq 1} \, D \Big (  q^{\star}_{j} \, \Big \| \, \textstyle{\frac{\sum_{k \neq 1}q^{\star}_{k}}{M-1} } \Big ) + \displaystyle{\sum\limits_{j \neq 2}} \, D \Big ( q^{\star}_{j} \, \Big \| \, \textstyle{\frac{\sum_{k \neq 2}q^{\star}_{k} }{M-1} }  \Big ) \nonumber \\
&\ =\  D(q^{\star}_{1} \, \| \, \mu) + \sum_{j \neq 2} D \Big (q^{\star}_{j} \, \Big \| \, \textstyle{ \frac{\sum_{k \neq 2}q^{\star}_{k} }{M-1} } \Big ) \nonumber  \\
& \hspace{0.2in} + \displaystyle{\sum_{j \neq 1}} \ \sum_{y \in \cY} \, q^{\star}_{j}(y) \log \bigg ( \frac{\frac{1}{M-1}\sum_{k \neq 1}q^{\star}_{k}(y)}{\pi}  \bigg ) \nonumber \\
&\ =\ D(q^{\star}_{1} \, \| \, \mu) + \sum_{j \neq 2} D \Big (q^{\star}_{j} \, \Big \| \, \textstyle{ \frac{\sum_{k \neq 2}q^{\star}_{k} }{M-1} } \Big ) \nonumber \\ 
& \hspace{0.2in} + (M-1) D \Big ( \textstyle{ \frac{\sum_{k \neq 1}q^{\star}_{k}}{M-1} } \, \Big \| \, \pi \Big )  \nonumber \\
&\ \geq\ D(q^{\star}_{1} \, \| \, \mu) + D \Big ( q^{\star}_{1} \, \Big \| \, \textstyle{\frac{\sum_{k \neq 2}q^{\star}_{k}}{M-1} } \Big ) \nonumber \\
&\ \geq\  2B \Big (\mu \, , \, \textstyle{ \frac{\sum_{k \neq 2} q^{\star}_{k}}{M-1} } \Big ) \nonumber \\
&\ =\  2 B \Big (\mu \, , \, \textstyle{ \frac{q^{\star}_{1}}{M-1} } +  \frac{M-2}{M-1}  \big ( \textstyle{ \frac{\sum_{k=3}^{M}q^{\star}_{k}}{M-2} } \big )  \Big ), \label{eqn-objectivefnc}
\end{align}
where the last inequality follows Lemma \ref{lm-2}.

On the other hand, it follows from (\ref{eqn-uppbd}) that the value on the right-side of (\ref{eqn-err-exp}), $V^{\star}$, satisfies 
\begin{align}
2B(\mu, \pi) &\ \geq\ V^{\star} \nonumber \\
&\ =\ D\left (q^{\star}_{1} \, \| \, \mu \right ) + \sum_{j \neq 1} D \left ( q^{\star}_{j} \, \| \, \pi \right ) \nonumber \\
&\ \geq\ \sum_{j=3}^{M} D \left (q^{\star}_{j} \, \| \, \pi \right ) \nonumber \\
&\ \geq\ (M-2) \, D \Big ( \textstyle{ \frac{1}{M-2}\sum_{k=3}^{M} q^{\star}_{k} } \, \Big \| \, \pi \Big ), \label{eqn-constraintfnc}
\end{align}
where the last inequality follows from the convexity of relative entropy.

Combining (\ref{eqn-objectivefnc}) and (\ref{eqn-constraintfnc}), we get that the value $V^{\star}$ on the right-side of (\ref{eqn-err-exp}) is lower bounded by
\begin{align}
\mathop{\min\limits_{q_{1} , q \, \in \, \mathcal{P}(\cY)}}_{(M-2)D(q \| \pi) \, \leq \, 2B(\mu, \pi)}  2B \Big  (\mu \, ,\,\textstyle{ \frac{1}{M-1}q_{1} + \frac{M-2}{M-1} q } \Big ). \label{eqn-lowerbd1}
\end{align}
Note that the constraint in (\ref{eqn-lowerbd1}) can be equally written as
\begin{align}
D \left( q_{1} \| \pi \right) 
+ (M-2) D \left( q \| \pi \right) \ \leq \ 
2B(\mu, \pi) + D \left( q_{1} \| \pi \right). \nonumber
\end{align}
Also by the convexity of relative entropy, it follows that
\begin{align}
D \left( q_{1} \| \pi \right)  + 
& (M - 2) D \left( q \| \pi \right) \ \geq \nonumber \\
&(M -1) D \Big ( \textstyle{ \frac{ q_{1}+\scriptsize{(M-2)}q}{\scriptsize{M-1}} } \Big \| \pi \Big ). \nonumber
\end{align}
As a result, the optimal value of (\ref{eqn-lowerbd1}) is further lower bounded by the optimal value of
\begin{align}
\mathop{\mathop{\min\limits_{q_{1} , q \, \in \, \mathcal{P}(\cY)}}_{(M-1)D \left( \scriptstyle{ \frac{1}{M-1}q_{1} + \frac{M-2}{M-1} q} \| \pi \right ) }}_{\leq\  2B(\mu, \pi)+D \left( q_{1} \| \pi \right)}  2B \Big  (\mu \, , \, \textstyle{ \frac{1}{M-1}q_{1} + \frac{M-2}{M-1} q } \Big ). \label{eqn-lowerbd2}
\end{align}
By the fact that $\pi$ has full support, it holds that
\begin{align}
D \left( q_{1} \| \pi \right) \ \leq \ -\log \Big ( \min\limits_{y \in \cY} \pi(y) \Big ) \ = \ C_{\pi} \ \leq \ \infty. \label{eqn-C}
\end{align} 
Proceeding from (\ref{eqn-lowerbd2}), by using (\ref{eqn-C}), we get that the optimal value of (\ref{eqn-err-exp}) is lower bounded by
\begin{align}
\mathop{\min\limits_{q' \, \in \, \mathcal{P}(\cY)}}_{D\left ( q' \| \pi \right) \, \leq \, \frac{1}{M-1}(2B(\mu, \pi)+ C_{\pi} )}  2\, 
B \left( \mu \, , \, q' \right). 
\label{eqn-lowerbd-inproof}
\end{align}

For any $\mu, \pi \in \mathcal{P}(\cY)$ with full supports, it holds that 
\begin{align}
\lim\limits_{M \rightarrow \infty} \hspace{0.02in} \frac{1}{M-1} \big (2B(\mu, \pi)+ C_{\pi} \big )\ = \ 0. \nonumber
\end{align} 
This and the continuity of $D \left(q \| \pi \right)$ in $q$ ($\pi$ has a full support) establish (\ref{eqn-limit}): the asymptotic optimality of \tmblue{the GL} test in the regime of large number of sequences. 
  
Furthermore, for any $\mu, \pi \in \mathcal{P}(\cY)$, $\mu \neq \pi$, the value of $\frac{1}{M-1}(2B(\mu, \pi)+ C( \pi ) )$ is strictly decreasing with $M$. Consequently, the feasible set in (\ref{eqn-lowerbd}) is nonincreasing with $M$, and hence the optimal value of (\ref{eqn-lowerbd}) is nondecreasing with $M$. 

\subsection{Proof of Proposition \ref{prop-4}}

The proposition follows as a special case of the second assertion of Theorem \ref{thm-10}, the proof of which is deferred to Appendix \ref{sec-prf-thm-10}.

\subsection{Proof of Theorem \ref{thm-4}}
We start by establishing universal consistency of the test under the null hypothesis. \tmblue{Applying the identity in (\ref{eqn-probability}) to the test statistics in (\ref{eqn-univ-detector-null}), it holds that}
\begin{align}
\mathbb{P}_{0}\{ \delta \neq 0\} 
&\ \leq\  \mathbb{P}_{0}\Big ( \cup_{j=1}^{M} \{ U_{j}^{\scriptsize{\mbox{univ}}} \geq \lambda_{n} \} \Big ) \nonumber \\
&\ \leq\  \sum_{j=1}^{M} \mathbb{P}_{0} \left \{ U_{j}^{\scriptsize{\mbox{univ}}} \geq \lambda_{n} \right \} \nonumber \\
&\ =\  M \mathbb{P}_{0} \left \{ U_{1}^{\scriptsize{\mbox{univ}}} \geq \lambda_{n} \right \}, \label{eqn-thm4-null-error-1}
\end{align}
where $U_{j}^{\scriptsize{\mbox{univ}}}$ is defined in (\ref{eqn-defU-univ}), and the last equality follows from the fact that all $\by^{(i)}$, $i=1, \ldots, M$, are identically distributed according to $\pi$.

We now proceed to bound $\mathbb{P}_{0}\{ U_{1}^{\scriptsize{\mbox{univ}}} \geq \lambda_{n} \}$ as follows:
\begin{align}
& \mathbb{P}_{0} \{ U_{1}^{\scriptsize{\mbox{univ}}}\geq \lambda_{n}\} \nonumber \\
= & \ \mathbb{P}_{0} \bigg \{ \sum_{j \neq 1} 
D \left (\gamma_{j} \Big \| 
\textstyle{\frac{\sum_{k \neq 1}\gamma_{k}}{M-1}} \right) \geq \lambda_{n} \bigg \} \nonumber \\
= & \ \mathbb{P}_{0} \bigg \{ \sum_{j \neq 1}D \left (\gamma_{j} \| \pi \right ) - (M-1)D \left (\textstyle{\frac{\sum_{k \neq 1}\gamma_{k}}{M-1}} \Big \| \pi \right ) \geq \lambda_{n} \bigg \} \nonumber \\
\leq & \ \mathbb{P}_{0} \bigg \{ \sum_{j \neq 1}D \left (\gamma_{j} \| \pi \right ) \geq \lambda_{n} \bigg \} \nonumber \\
\leq & \ \mathbb{P}_{0} \left( 
\cup_{j \neq 1} 
\left \{
D \left (\gamma_{j} \| \pi \right ) \geq \frac{1}{M-1}\lambda_{n} 
\right \}
\right ) \nonumber \\
\leq & \ (M-1)\mathbb{P}_{0} \Big \{ D \left (\gamma_{2} \| \pi \right ) \geq \frac{1}{\textstyle{M-1}} \lambda_{n} \Big \}, \label{eqn-thm4-null-error-2}
\end{align} 
where the first inequality follows from the non-negativity of the relative entropy, and the last inequality follows from the fact that all $\by^{(j)}$, $j \neq 1$, are identically distributed according to $\pi$.
By the fact that the set of all possible empirical distributions of $\left( y_1, \ldots, y_n \right)$ is upper bounded by $\left( n + 1 \right)^{\vert \mathcal{Y} \vert}$ (cf. \cite{cove-thom-eit-book-2006}[Theorem 11.1.1]), and (\ref{eqn-probability}), we get that  
\begin{align}
\mathbb{P}_{0}  \Big \{  & D \left (\gamma_{2} \| \pi \right ) \geq \frac{1}{M-1}\lambda_{n} \Big \} \nonumber \\ 
& \ \ \leq\ (n+1)^{\vert \mathcal{Y} \vert} \exp(-\frac{n}{M-1} \lambda_{n}). \label{eqn-thm4-null-bnd}
\end{align} 
It then follows from (\ref{eqn-thm4-null-error-1}), (\ref{eqn-thm4-null-error-2}) and (\ref{eqn-thm4-null-bnd}) that
\begin{align}
\mathbb{P}_{0}\{ \delta \neq 0\} \leq  M^{2}\exp \Big  \{-\frac{n}{M-1} \lambda_{n}+\vert \mathcal{Y} \vert \log(n+1) \Big \}. \label{eqn-thm4-null-bnd-1}
\end{align}
By choosing $\lambda_{n} = 2 (M-1) \vert \mathcal{Y} \vert
\frac{\log{\left( n+ 1\right)}}{n},$ we get from (\ref{eqn-thm4-null-bnd-1}) that
\begin{align}
\lim_{n \rightarrow \infty} \mathbb{P}_{0}\{ \delta \neq 0 \} \ =\ 0. \nonumber
\end{align}

Next we treat the exponent for the conditional probability of error under every non-null hypothesis.  In particular, by the symmetry of the test (\ref{eqn-univ-detector-null}) among all the $M$ non-null hypotheses, it suffices to consider the conditional error probability under just the first hypothesis, which can be upper bounded as follows:
\begin{align}
\hspace{-0.09in} \mathbb{P}_{1}\left \{ \delta \neq 1 \right \} 
&\ \leq \ \mathbb{P}_{1}
\left ( 
\cup_{j \neq 1} 
\Big \{ U_{1}^{\scriptsize{\mbox{univ}}} \ \geq\ U_{j}^{\scriptsize{\mbox{univ}}} - \lambda_{n} \Big \}
\right ) \nonumber \\
&\ \leq \ \sum_{j \neq 1} 
\mathbb{P}_{1}  
\Big \{ U_{1}^{\scriptsize{\mbox{univ}}} \ \geq\ U_{j}^{\scriptsize{\mbox{univ}}} - \lambda_{n} \Big \}
\nonumber \\
&\ \leq \ (M-1)\mathbb{P}_{1}  
\Big \{ 
U_{1}^{\scriptsize{\mbox{univ}}} \ \geq\ U_{2}^{\scriptsize{\mbox{univ}}} - \lambda_{n} \Big \}. 
\label{eqn-thm4-unionbd-1}
\end{align}
For an arbitrary $\lambda_{0} > 0$, as $\lambda_{n} \rightarrow 0,$ it holds that $\lambda_{n} \leq \lambda_{0}$ for $n$ sufficiently large and hence that
\begin{align}
\mathbb{P}_{1}  
\Big \{ 
U_{1}^{\scriptsize{\mbox{univ}}}  \geq  U_{2}^{\scriptsize{\mbox{univ}}} - \lambda_{n} 
\Big \} 
\ \leq\  
\mathbb{P}_{1}  \Big \{ 
U_{1}^{\scriptsize{\mbox{univ}}} \ \geq\ U_{2}^{\scriptsize{\mbox{univ}}} - \lambda_{0} \Big \}. \label{eqn-thm4-notnull}
\end{align}
The exponent of the right-side of (\ref{eqn-thm4-notnull}) can be computed by applying Lemma \ref{lm-1} with $J=M$, $p_{1}=\mu$, $p_{j}=\pi$, $j=2, \ldots, M$ and \tmblue{(cf.(\ref{eqn-defU-univ}))}
\begin{align}
E (\lambda_{0}) \ \triangleq \ \bigg \{ 
\left( q_{1}, \ldots, q_{M} \right) & \,: \,
\sum_{j \neq 1} \ D \Big ( q_{j} \, \Big \| \, \textstyle{\frac{\sum_{k \neq 1} q_{k}}{M-1} } \Big )
\nonumber \\
&\hspace{-0.2in} \geq \  
\displaystyle{\sum_{j \neq 2}} \ D \Big (q_{j} \, \Big \| \, 
\textstyle{ \frac{\sum_{k \neq 2} q_{k}}{M-1} } \Big ) - \lambda_{0}
\bigg \} \nonumber
\end{align}
to be
\begin{align}
\mathop{\min\limits_{\left ( q_{1}, \ldots, q_{M} \right ) \in E(\lambda_{0})}} D \left ( q_{1} \| \mu \right ) + D \left ( q_{2} \| \pi \right )+ \ldots +D\left ( q_{M} \| \pi \right ). \label{eqn-thm4-optimization}
\end{align}
Since $\lambda_{0}$ can be arbitrarily close to zero, the exponent for the left-side of (\ref{eqn-thm4-notnull}) is lower bounded by
\begin{align}
\lim_{\lambda_{0} \rightarrow 0} \ \mathop{\min\limits_{\left ( q_{1}, \ldots, q_{M} \right )\, \in \, E(\lambda_{0})}} D \left ( q_{1} \| \mu \right ) &+ D \left ( q_{2} \| \pi \right ) \nonumber \\
& + \ldots +D\left ( q_{M} \| \pi \right ). \nonumber
\end{align}
Let 
\begin{align}
E  \ \triangleq \ \bigg \{ 
\left( q_{1}, \ldots, q_{M} \right) \,: \,
\sum_{j \neq 1} \ & D \Big ( q_{j} \, \Big \| \, \textstyle{\frac{\sum_{k \neq 1} q_{k}}{M-1} } \Big )
\nonumber \\
& \geq \  
\displaystyle{\sum_{j \neq 2}} \ D \Big (q_{j} \, \Big \| \, 
\textstyle{ \frac{\sum_{k \neq 2} q_{k}}{M-1} } \Big ) 
\bigg \}. \nonumber
\end{align}
By the fact that $E(\lambda_{0})$ is closed and compact for any $\lambda_{0} > 0$, and that the objective function in (\ref{eqn-thm4-optimization}) is continous, the exponent for the left-side of (\ref{eqn-thm4-notnull}) is lower bounded by 
\begin{align}
\mathop{\min\limits_{\left ( q_{1}, \ldots, q_{M} \right )\, \in \, E}} D \left ( q_{1} \| \mu \right ) + D \left ( q_{2} \| \pi \right )+ \ldots +D\left ( q_{M} \| \pi \right ),  \label{eqn-thm4-optimization-1}
\end{align}
\tdblue{as required.}

\subsection{Proof of Proposition \ref{prop-6}}

The \tdblue{proposition} follows from a well-known result in detection and estimation in the context of multihypothesis testing problem\cite{tsit-mcss-1988}.  In particular, the optimal error exponent for testing $M$ hypotheses with i.i.d. observations with respect to $p_1, p_2, \ldots, p_M$ is characterized as
$\min\limits_{1 \leq i < j \leq M}  C \left( p_i, p_j \right)$.

When all the $\left \{ \mu_i \right \}_{i=1}^M$ and $\pi$ are known, the underlying outlier hypothesis testing problem is just a multihypothesis testing problem based on i.i.d. vector observations (with $M$ independent components) and consequently, the optimal error exponent can be computed as
\begin{align} 
&\hspace{-0.1in} \min_{S \neq S'} 
C \bigg(
\prod_{i \in S} \mu_i \left( y_i \right) \prod_{j \notin S} \pi \left( y_j \right), 
\prod_{i \in S'} \mu_i \left( y_i \right) \prod_{j \notin S'} \pi \left( y_j \right) 	
\bigg) \nonumber \\
&\hspace{-0.12in} =  \min_{S \neq S'} 
C \Big( \prod_{i \in S \backslash S'} \mu_i \left( y_i \right) 
\prod_{j \in S' \backslash S} \pi \left( y_j \right), \nonumber \\ 
  & \hspace{1.2in} \prod_{i \in S \backslash S'} \pi \left( y_i \right) 	
  \prod_{j \in S' \backslash S} \mu_j \left( y_j \right) \Big) \nonumber \\
&\hspace{-0.12in} =
\min_{S \neq S'} \max_{s \in [0, 1]}
-\log \Bigg [ \mathop{\sum_{y_i,\ i\in S \backslash S'}}_{y_j,\ j\in S' \backslash S}
 \hspace{-0.08in} \bigg (\prod_{i \in S \backslash S'} 
 \hspace{-0.04in} \mu_i \left( y_i \right)^{1-s} \pi \left( y_i \right)^{s} \nonumber \\
  & \hspace{1in} \cdot  \prod_{j \in S' \backslash S} 
   \pi \left( y_j \right)^{1-s} \mu_j \left( y_j \right)^{s} \bigg ) \Bigg ]	\label{eqn-pf-prop7-1} \\
&\hspace{-0.12in} = \min_{1 \leq i < j \leq M}\ \max_{s \in [0, 1]}\ 
 -\log \Bigg [ \sum_{y_i, y_j} 
 \Big ( \mu_i \left( y_i \right)^{1-s} \pi \left( y_i \right)^{s}  \nonumber \\
  &\hspace{1.2in}   \cdot  \pi \left( y_j \right)^{1-s} \mu_j \left( y_j \right)^{s} \Big ) \Bigg ] 
  \label{eqn-pf-prop7-2} \\
&\hspace{-0.12in}=\ 
 \min_{1 \leq i < j \leq M}\  
 C \left(
 \mu_i \left( y \right) \pi \left( y'\right), 
 \pi \left( y \right) \mu_j \left( y' \right)
 \right),	\nonumber
\end{align}
\tdblue{where the equality in (\ref{eqn-pf-prop7-2}) follows by virtue of fact that the outer minimum in (\ref{eqn-pf-prop7-1}) is attained among the pairs of $S, S',$ with the largest number of \tmblue{sequences} in their intersections: $T-1$.} 

When all the \tmblue{outliers} are identically distributed, i.e., $\mu_{i} = \mu$, $i=1, \dots, M$, \tdblue{this optimal error exponent can be further simplified to be}
\begin{align} 
 & \min_{1 \leq i < j \leq M}\  
 C \left(
 \mu_i \left( y \right) \pi \left( y'\right), 
 \pi \left( y \right) \mu_j \left( y' \right)
 \right) \nonumber \\
&=\  C \left(
 \mu \left( y \right) \pi \left( y'\right), 
 \pi \left( y \right) \mu \left( y' \right)
 \right) \ =\  2B(\mu, \pi). \label{eqn-ident-chernof} 
\end{align}

\subsection{Proof of Theorem \ref{thm-7}}

\tmblue{For each $S \subset \cal{S},$ denote the test statistic in (\ref{eqn-univ-mult-detector-pi'}) as 
\begin{align}
U_{S}^{\mbox{\scriptsize{typ}}} \triangleq \sum\limits_{j \notin S} D(\gamma_{j} \| \pi). \label{eqn-defU-typ-multi}
\end{align}}
\tmblue{Consider the test $\delta$ in (\ref{eqn-univ-mult-detector-pi}) and (\ref{eqn-univ-mult-detector-pi'}).}
It follows from the fact that for every $S \in \cal{S},$
\begin{align}
\mathbb{P}_S \left \{ \delta \neq S \right \}
\ = \ \mathbb{P}_S \left \{ 
\mathop{\cup}_{S' \neq S} 
\left \{  
U_{S}^{\mbox{\scriptsize{typ}}} \geq U_{S'}^{\mbox{\scriptsize{typ}}}
\right \}
 \right \}
 \nonumber
\end{align}
that 
\begin{align}
& \max\limits_{S \neq S'}
\ \mathbb{P}_S \Big \{ 
U_{S}^{\mbox{\scriptsize{typ}}} \geq U_{S'}^{\mbox{\scriptsize{typ}}}
\Big \} \nonumber \\
 &\ \leq \max\limits_{S \in \cal{S}}  \mathbb{P}_S \left \{ \delta \neq S \right \} \nonumber \\
&\ \leq \max\limits_{S \in \cal{S}} \sum_{S' \neq S}
\mathbb{P}_S \left \{ 
U_{S}^{\mbox{\scriptsize{typ}}} \geq U_{S'}^{\mbox{\scriptsize{typ}}}
\right \}  \nonumber \\
&\ \leq
\left( \vert \mathcal{S} \vert  - 1 \right)
\max\limits_{S \neq S'}
\mathbb{P}_S \left \{ 
U_{S}^{\mbox{\scriptsize{typ}}} \geq U_{S'}^{\mbox{\scriptsize{typ}}}
\right \}.
\label{eqn-pfThm1-bd-pwerr}
\end{align}
Next, we get from (\ref{eqn-defU-typ-multi}) that for any $S \neq S' \in \cal{S},$
\begin{align}
\mathbb{P}_S \left \{ 
U_{S}^{\mbox{\scriptsize{typ}}} \geq U_{S'}^{\mbox{\scriptsize{typ}}}
\right \}
= 
\mathbb{P}_S \left \{
\sum\limits_{i \notin S} D(\gamma_i \| \pi) \geq 
\sum\limits_{i \notin S'} D(\gamma_i \| \pi)
\right \}.
\nonumber
\end{align}
Applying Lemma \ref{lm-1} with $J=M,\ p_i = \mu_i,\ i \in S,\ p_j = \pi,\ j \notin S,$ and
\begin{align}
E = \left \{ \left( q_1, \ldots, q_M \right):~ 
\sum\limits_{i \notin S} D(q_i \| \pi)  \geq
\sum\limits_{i \notin S'} D(q_i \| \pi)
\right \},
\end{align}
we get that the exponent for $~\mathbb{P}_S \left \{
U_{S}^{\mbox{\scriptsize{typ}}} \geq U_{S'}^{\mbox{\scriptsize{typ}}}
\right \}~$ is given by the value of the following optimization problem
\tmred{
\begin{align}
\mathop{\min\limits_{\left \{q_i \right \}_{i \in S \backslash S'},\  
\left \{q_j \right \}_{j \in S' \backslash S} }}\, \sum_{i \in S \backslash S'} D \left ( q_i \| \mu_i \right ) 
+ \hspace{-0.05in} \sum_{j \in S' \backslash S} D \left ( q_j \| \pi \right ), 
\label{eqn-pfThm1-pwerr}
\end{align}
where the minimum above is over the set of $\left \{q_i \right \}_{i \in S \backslash S'},\  
\left \{q_j \right \}_{j \in S' \backslash S}$, such that
\begin{align}
\sum\limits_{j \in S' \backslash S} D(q_j \| \pi) \ \geq\ 
\sum\limits_{i \in S \backslash S'} D(q_i \| \pi).  \nonumber
\end{align}}
We now show that the optimum value in (\ref{eqn-pfThm1-pwerr}) is equal to 
$\sum\limits_{i \in S \backslash S'}  2 B \left( \mu_i, \pi \right).$  First, we show that the latter is a lower bound for the former.  Substituting the constraint in (\ref{eqn-pfThm1-pwerr}) into the objective function, we get that the value of (\ref{eqn-pfThm1-pwerr}) is lower bounded by
\begin{align}
\hspace{-0.02in} \mathop{\min\limits_{\left \{q_i \right \}_{i \in S \backslash S'} }}
\hspace{-0.02in} \sum_{i \in S \backslash S'} 
\hspace{-0.03in} D \left ( q_i \| \mu_i \right ) + D \left ( q_i \| \pi \right ) 
= \hspace{-0.06in} \sum_{i \in S \backslash S'} 
\hspace{-0.04in} 2 B \left( \mu_i, \pi \right), 
\label{eqn-pfThm1-pwerr-2}
\end{align}
where the equality follows from Lemma \ref{lm-2}.  Second, note that $\vert S \backslash S' \vert$ is always equal to $\vert S' \backslash S \vert$, and, hence, we can make a suitable correspondence between elements of $S \backslash S'$ to those of $S' \backslash S$.  The converse implication now follows by assigning for every $i \in S \backslash S',$ and the corresponding $j \in S' \backslash S, \  q_i = q_j = 
\frac{\mu_i \left( y \right)^{1/2}  \pi \left( y \right)^{1/2}}
{\sum\limits_{y' \in \cY} \mu_i \left( y' \right)^{1/2}  \pi \left( y' \right)^{1/2}},$ 
and note that this assignment trivially satisfies the constraint in (\ref{eqn-pfThm1-pwerr}) and gives rise to the objective function being equal to 
$\sum\limits_{i \in S \backslash S'}  2 B \left( q_i, \pi \right)$.

Lastly, it follows from (\ref{eqn-pfThm1-bd-pwerr}) that 
\begin{align}
&\lim\limits_{n \rightarrow \infty}
- \frac{1}{n}
\log{ \left(
\max\limits_{S \in \cal{S}}  \mathbb{P}_S \left \{ \delta \neq S \right \}
\right)} 	\nonumber \\
&=\ 
\min\limits_{S \neq S'}\ \sum_{i \in S \backslash S'}  2 B \left( \mu_i, \pi \right)
\ =\ \min_{1 \leq i \leq M}\  2B \left( \mu_i,  \pi  \right).	\nonumber
\end{align}

When $\mu_{i}=\mu$, $i=1, \ldots, M$, 
\begin{align}
\min_{1 \leq i \leq M}\  2B \left( \mu_i,  \pi  \right) = 2B(\mu, \pi). \nonumber
\end{align}

\subsection{Proof of Theorem \ref{thm-8}}
\tmblue{For each $S \subset \cal{S},$ denote the test statistic in (\ref{eqn-univ-mult-detector'}) as
\begin{align}
U_{S}^{\mbox{\scriptsize{univ}}} \triangleq \sum\limits_{j \notin S}
D \left( \gamma_{j} \,\big \| \, \textstyle{\frac{\sum_{k \notin S} \gamma_{k}}{M-T}} \right). \nonumber
\end{align}
Consider the test $\delta$ specified by (\ref{eqn-univ-mult-detector}) and (\ref{eqn-univ-mult-detector'}).} 
It now follows in the manner similar to (\ref{eqn-pfThm1-bd-pwerr}) that 
\begin{align}
&\max\limits_{S \neq S'} \
\mathbb{P}_S \left \{ 
U_{S}^{\mbox{\scriptsize{univ}}} \geq U_{S'}^{\mbox{\scriptsize{univ}}}
\right \} \nonumber \\
& \leq\ \max\limits_{S \in \cal{S}}  \mathbb{P}_S \left \{ \delta \neq S \right \} \nonumber \\
& \leq\ \max\limits_{S \in \cal{S}} \sum_{S' \neq S}
\mathbb{P}_S \left \{ 
U_{S}^{\mbox{\scriptsize{univ}}} \geq U_{S'}^{\mbox{\scriptsize{univ}}}
\right \}  \nonumber \\
& \leq\ 
\left( \vert \mathcal{S} \vert  - 1 \right)
\max\limits_{S \neq S'}
\mathbb{P}_S \left \{ 
U_{S}^{\mbox{\scriptsize{univ}}} \geq U_{S'}^{\mbox{\scriptsize{univ}}}
\right \}.
\label{eqn-pfThm2-bd-pwerr}
\end{align}

The assertion (\ref{eqn-thm6-err-exp}) now follows from (\ref{eqn-pfThm2-bd-pwerr}) upon noting that the application of Lemma \ref{lm-1} with $J=M,\ p_i = \mu_i,\ i \in S,\ p_j = \pi,\ j \notin S,$ and
\begin{align}
E \ =\ \bigg \{ \left( q_1, \ldots, q_M \right):\ 
\sum\limits_{i \notin S} 
D \left( q_i \,\Big \| \, \textstyle{\frac{\sum_{k \notin S} q_k}{M-T}} \right) \nonumber \\
\ \geq\ 
\sum\limits_{i \notin S'} 
D\left( q_i  \,\Big \| \, \textstyle{\frac{\sum_{k \notin S'} q_k}{M-T}} \right)
\bigg \},
\label{eqn-PfThm2-constraint}
\end{align}
gives that the exponent for $~\mathbb{P}_S \left \{
U_{S}^{\mbox{\scriptsize{univ}}} \geq U_{S'}^{\mbox{\scriptsize{univ}}}
\right \}~$ is equal to \tmred{the value of the following optimization problem
\begin{align}
\min\limits_{q_{1}, \ldots, q_{M}} \sum_{i \in S} D \left ( q_i \| \mu_i \right ) 
+ \sum_{j \notin S} D \left ( q_j \| \pi \right ),
\label{eqn-pfThm2-pwerr}
\end{align}
where the minimum is over the set of $\{q_{1}, \ldots, q_{M}\}$ such that
\begin{align} 
\sum\limits_{i \notin S} D\left( q_i \, \Big \| \,
\textstyle{\frac{\sum_{k \notin S} q_k}{M-T}}  \right) 
\ \geq\ 
\sum\limits_{i \notin S'} D\left( q_i \, \Big \| \,
\textstyle{\frac{\sum_{k \notin S'} q_k}{M-T}}  \right). \nonumber
\end{align}}

Lastly, the assertion of universally exponential consistency of \tmblue{the GL test in (\ref{eqn-univ-mult-detector}) and (\ref{eqn-univ-mult-detector'})} follows from the compactness of the the feasible set of (\ref{eqn-pfThm2-pwerr}), continuity of the objective function in (\ref{eqn-pfThm2-pwerr}), and 
the fact that the objective function of (\ref{eqn-pfThm2-pwerr}) can only be zero at a collection $\left( q_i = \mu_i, i \in S,\ q_j = \pi, j \notin S \right),$ which is not in the constraint set.

\subsection{Proof of Theorem \ref{thm-9}}

First let denote the minimizing $S$ and $S'$ in the outer minimum of (\ref{eqn-thm6-err-exp}) by $S^{\star}$ and ${S'}^{\star}$ respectively, and the minimizing tuple $q_1, \ldots, q_M$ in the inner minimum of (\ref{eqn-thm6-err-exp}) by $q_1^*, \ldots, q_M^*$.  Then, we get that the achievable error exponent in (\ref{eqn-thm6-err-exp}) is lower bounded as
\begin{align}
&\ \geq\  
\sum\limits_{i \in S^{\star}} D \left( q^{\star}_{i} \| \mu_i \right) + 
\sum\limits_{j \notin S^{\star}} D \left( q^{\star}_{j} \|  \pi \right) \nonumber \\
& \hspace{0.2in} - \sum_{j \notin S^{\star}} 
D \left (  q^{\star}_{j} \, \Big \| \,
\textstyle{ \frac{\sum_{k \notin S^{\star}} q^{\star}_{k}}{M-T} }
\right ) + 
\displaystyle{ \sum\limits_{j \notin {S'}^{\star} } } 
D \left ( q^{\star}_{j} \, \Big \| \,
\textstyle{ \frac{\sum_{k \notin {S'}^{\star} }q^{\star}_{k}} {M-T} }  
\right )
\nonumber  \\
&\ =\  \sum\limits_{i \in S^{\star}} D \left( q^{\star}_{i} \, \| \, \mu_i \right) 
+ \sum_{j \notin {S'}^{\star}} 
D \left (q^{\star}_{j}  \,\Big \| \,
\textstyle{ \frac{\sum_{k \notin {S'}^{\star}} q^{\star}_{k}}{M-T} } 
\right ) \nonumber \\
&\hspace{0.2in} + \left( M- T \right)
D \left (
\textstyle{ \frac{\sum_{k  \notin S^{\star}} q^{\star}_{k}}{M-T} } 
\, \Big \| \, \pi 
\right ) 
\nonumber \\
&\ \geq\ D \left( q^{\star}_{t} \, \| \, \mu_t \right) + 
\tdblue{D \left ( q^{\star}_t \, \Big \| \, 
\textstyle{\frac{\sum_{k \notin {S'}^{\star}} q^{\star}_{k}}{M-T} } 
\right )} 
\nonumber \\
&\ \geq\ 
\tdblue{2B \left (\mu_t \, , \, 
\textstyle{\frac{\sum_{k \notin {S'}^{\star}} q^{\star}_{k}}{M-T} } 
\right ),} 
\label{eqn-PfThm3-lowerbd-1}
\end{align}
where $t$ is an arbitrarily chosen element in $S^{\star} \backslash {S'}^{\star}.$

On the other hand, it follows from Proposition \ref{prop-6} that 
\begin{align}
& \min\limits_{1\leq i < j \leq M} C \left( \mu_i (y) \pi(y'), \pi(y)\mu_{j}(y') \right) \nonumber \\
& \geq\  \sum\limits_{i \in S^{\star}} 
D\left (q^{\star}_{i} \, \| \, \mu_i \right) 
+ 
\sum\limits_{j \notin S^{\star}} D \left ( q^{\star}_{j} \, \| \, \pi \right ) \nonumber \\
&\geq\ \sum_{j \notin S^{\star} \cup {S'}^{\star}} 
D \left (q^{\star}_{j} \, \| \, \pi \right ) \nonumber \\
&\geq\ (M-T-\vert S^{\star} \backslash {S'}^{\star} \vert) \, 
\tdblue{D \left ( 
\textstyle{ 
\frac{\sum_{j \notin S^{\star} \cup {S'}^{\star}} ~q^{\star}_{j}}
{\left( M-T-\vert S^{\star} \backslash {S'}^{\star} \vert \right)}  } 
\, \Big \| \, \pi 
\right ).} \label{eqn-PfThm3-lowerbd-2}
\end{align}
It now follows from (\ref{eqn-PfThm3-lowerbd-2}) that
\begin{align}
& (M-T) \, 
D \left ( 
\textstyle{ \frac{\sum_{k \notin {S'}^{\star} } q^{\star}_{k}}{M-T} } 
\, \Big \| \, \pi 
\right ) \nonumber \\
&\leq\ 
(M-T-\vert S^{\star} \backslash {S'}^{\star} \vert) \, 
\tdblue{D \left ( 
\textstyle{ 
\frac{\sum_{j \notin S^{\star} \cup {S'}^{\star}} ~q^{\star}_{j} }
{ \left( M-T-\vert S^{\star} \backslash {S'}^{\star} \vert \right)}
} 
\, \Big \| \, \pi 
\right )}  
\nonumber \\	
&\ \ \ \ \ +\ 
(\vert S^{\star} \backslash {S'}^{\star} \vert) \, 
\tdblue{
D \left ( 
\textstyle{ 
\frac{\sum_{i \in S^{\star} \backslash {S'}^{\star}} ~q^{\star}_i}
{\vert S^{\star} \backslash {S'}^{\star} \vert}
} 
\, \Big \| \, \pi 
\right )	
}
\nonumber \\
&\leq \min\limits_{1\leq i < j \leq M} C \left( \mu_i (y) \pi(y'), \pi(y)\mu_{j}(y') \right) 
\ +\ \vert S^{\star} \backslash {S'}^{\star} \vert C_{\pi} 
\nonumber \\
&\leq \min\limits_{1\leq i < j \leq M} C \left( \mu_i (y) \pi(y'), \pi(y)\mu_{j}(y') \right)  \ +\ T C_{\pi}. 
\label{eqn-PfThm3-lowerbd-3}
\end{align}
The lower bound in (\ref{eqn-thm9-lowerbd}) now follows from (\ref{eqn-PfThm3-lowerbd-1}) and (\ref{eqn-PfThm3-lowerbd-3}).

The assertion (\ref{eqn-thm9-limit}) now follows from (\ref{eqn-thm9-lowerbd}), Proposition \ref{prop-6} and the continuity of $B \left( \mu, q \right)$ and $D \left( q \| \pi \right)$ in the argument $q$.  \tdblue{The assertion (\ref{eqn-thm9-eff-iden-outliers}) follows as a special case of (\ref{eqn-thm9-limit}).}

It is now left only to prove the asymptotically exponential consistency of the test.
Having proved \tdblue{(\ref{eqn-thm9-limit}), this assertion} now follows upon noting that for every $i, j,\ 1 \leq i < j \leq M,$ it holds that
\begin{align}
&C \left(
 \mu_i \left( y \right) \pi \left( y'\right), 
 \pi \left( y \right) \mu_j \left( y' \right)
 \right) \nonumber \\
&\leq\ 
 2 B \left(
 \mu_i \left( y \right) \pi \left( y'\right), 
 \pi \left( y \right) \mu_j \left( y' \right)
 \right) \nonumber \\
& =\  
- 2\log \bigg ( 
\sum_{y, y' \in \cY \times \cY} 
\left( \mu_i \left( y \right) \pi \left( y' \right) \right)^{\frac{1}{2}} 
\left( \pi \left( y \right) \mu_j \left( y' \right)  \right)^{\frac{1}{2}}
\bigg ) \nonumber \\
& =\ 
2B \left( \mu_i, \pi \right) + 2B \left( \mu_j, \pi \right),		\nonumber
\end{align}
where the first inequality above follows from Lemma \ref{lm-3}.

\subsection{Proof of Theorem \ref{thm-10}} \label{sec-prf-thm-10}

We first prove that for every hypothesis set excluding the null hypothesis, \tmblue{the GL test in (\ref{eqn-univtestsecB1})} is universally exponentially consistent.

\tmblue{For each $S \subset \cal{S},$ denote the test statistic in (\ref{eqn-univtestsecB1'}) as 
\begin{align}
\bar{U}_{S}^{\mbox{\scriptsize{univ}}} \triangleq \sum_{i \in S} D \big (\gamma_{i} \big  \| \textstyle \frac{\sum_{k \in S} \gamma_{k}}{T} \big )+ \displaystyle
\sum_{j \notin S}D \big (\gamma_{j} \big \| \textstyle \frac{\sum_{k \notin S} \gamma_{k}}{M-T} \big ). \nonumber
\end{align}
}
Following the same argument leading to (\ref{eqn-pfThm1-bd-pwerr}), it suffices to show that for any $S, S' \in \mathcal{S}, S' \neq S$,
\begin{align}
\lim_{n \rightarrow \infty} - \frac{1}{n}\log \Big ( \mathbb{P}_{S} \left \{ 
\bar{U}_{S}^{\mbox{\scriptsize{univ}}} \ \geq \ \bar{U}_{S'}^{\mbox{\scriptsize{univ}}} 
\right \}
 \Big ) \ > \ 0.
\end{align} 

Applying Lemma \ref{lm-1} with $J=M,\ p_i = \mu,\ i \in S,\ p_j = \pi,\ j \notin S,$ andf
\begin{align}
E_{(S, S')} & =\Big \{  
\left( q_1, \ldots, q_M \right):\ 
\sum_{i \in S} 
D \left (
q_i  \,\Big \|\, \textstyle{\frac{\sum_{k \in S} q_k}{T}}
\right ) \nonumber \\
& +  \sum\limits_{j \notin S}  
D \left (q_j \Big \|  \textstyle{\frac{\sum_{k \notin S} q_{k}}
{ M - T }}   \right )  \geq
\sum_{i \in S'} 
D \left ( q_i  \, \Big \| \,  \textstyle{\frac{\sum_{k \in S'} q_{k}}{T}}
\right ) \nonumber \\
&\hspace{1.2in} +  
\sum\limits_{j \notin S'}  
D \left ( q_j \,\Big \|\,  \textstyle{\frac{\sum_{k \notin S'} q_k}{M - T}}   
\right )
\Big \},
\nonumber
\end{align}
we get that
the exponent for $\mathbb{P}_{S} \left \{ \bar{U}_{S}^{\mbox{\scriptsize{univ}}} \ \geq \ \bar{U}_{S'}^{\mbox{\scriptsize{univ}}} \right \}$ is given by the value of the following optimization problem
\begin{align}
\min\limits_{\{q_{1}, q_{2}, \ldots, q_{M}\} \in E_{(S, S')}} \ \sum_{i \in S} D(q_{i} \| \mu) + \sum_{j \notin S} D(q_{j} \| \pi). \label{eqn-optimB}
\end{align}
The solution to $\sum\limits_{i \in S} D \left( q_{i} \| \mu \right) + 
\sum\limits_{j \notin S} D \left( q_{j} \| \pi \right) =0$ is uniquely given by $q_{i} =\mu$ for $i \in S$, $q_{j}= \pi$ for $j \notin S$. Because $\vert S \vert < M/2$, $\vert S' \vert < M/2,$ there is no $S, S' \in \mathcal{S},$ such that $S=\{1, 2, \ldots, M\} \setminus S'$. Let $q_{i} =\mu$ for $i \in S$, $q_{j}= \pi$ for $j \notin S$, it then follows that
\begin{align}
0 \, & =\, 
\sum_{i \in S} 
D \left ( q_{i} \,\Big \|\, \textstyle{\frac{\sum_{k \in S} q_{k}}{T}} \right ) + 
\sum\limits_{j \notin S}  
D \left (q_{j} \,\Big \|\,  \textstyle{\frac{\sum_{k \notin S} q_{k}}{ M - T }} \right ) \nonumber \\
& < \sum_{i \in S'} 
D \left ( q_i \,\Big \|\,  \textstyle{\frac{\sum_{k \in S'} q_{k}}{T}} \right ) + 
\sum\limits_{j \notin S'}  
D \left ( q_j \,\Big \|\,  \textstyle{\frac{\sum_{k \notin S'} q_{k}}{ M - T}}   \right ) 
\nonumber
\end{align} 
for any $S, S' \in \cal{S}$, $S' \neq S$. In other words, the objective function in (\ref{eqn-optimB}) is strictly positive at any feasible $(q_{1}, q_{2}, \ldots, q_{M})$. By the continuity of the objective function in (\ref{eqn-optimB}) and the fact that $E_{(S, S')}$ is compact for any $S, S' \in \cal{S}$, it holds that the value of the optimization function in (\ref{eqn-optimB}) is strictly positive for every pair of $S, S' \in \cal{S}$, $S \neq S'$. This establishes the exponential consistency of \tmblue{the GL test in (\ref{eqn-univtestsecB1})}.

\tdblue{Next to prove the second assertion, it suffices to prove that even when the typical distribution is known, there cannot exist a universally exponentially consistent test in differentiating the null hypothesis from any other hypothesis with a positive number of outliers.}  \tdblue{To this end,} let $S \subset \{1, 2, \ldots, M\}$, $\vert S\vert \geq 1$ denote an arbitrary set of \tmblue{outliers}. To distinguish between the null hypothesis and $S$, \tdblue{a test is done} based on a decision rule $\delta : \cY^{Mn} \rightarrow \{0, 1\}$,  where $0$ corresponds to the null hypothesis and $1$ the hypothesis with $S$ being the outliers. It should be noted that $\delta$ can only be a function of $\pi$ and the observations $\cY^{Mn}$. 

We first show that in order to distinguish between the null hypothesis and $S$, the empirical distributions of all the \tmblue{sequences} $\gamma_1, \ldots, \gamma_M$ and $\pi$ are sufficient statistics for the error exponent.  In particular, we now show that given any test, there is another test that achieves the same error exponent with its decision being made based {\em only on} the empirical distributions of all $M$ \tmblue{sequences} and $\pi$.  To this end, for feasible empirical distributions (for certain $n$) $\gamma_1, \dots, \gamma_M,$ let us denote the set of all M sequences conforming to these empirical distributions by $T_{\left( \gamma_1, \dots, \gamma_M \right)}$.  Among these observation sequences, let us denote the set of M sequences for which $\delta$ decides for the null hypothesis by $T_{\left( \gamma_1, \dots, \gamma_M \right)}^{0, \pi},$ which may depend on $\pi$.  Now consider another test $\delta'$ which decides on one of the two hypotheses based only on $\gamma_1, \dots, \gamma_M$ and $\pi$.  Specifically, this new test is such that for all M sequences with empirical distributions $\gamma_1, \dots, \gamma_M$, it decides for the null hypothesis if  $\vert T_{\left( \gamma_1, \dots, \gamma_M \right)}^{0, \pi} \vert \geq \frac{1}{2} \vert T_{\left( \gamma_1, \dots, \gamma_M \right)} \vert,$ and for $S$ otherwise.  It follows from this construction of $\delta'$ that for any $\mu$ and $\pi,$ 
\begin{align}
	& \max{\left( \mathbb{P}_0 \left \{ \delta' \neq 0 \right \}, 	
	             \mathbb{P}_1 \left \{ \delta' \neq 1 \right \}\right)} \nonumber \\
	 & \ \  \leq\ 
	2 \max{\left( \mathbb{P}_0 \left \{ \delta \neq 0 \right \}, 	
	             \mathbb{P}_1 \left \{ \delta \neq 1 \right \}\right)},
\nonumber
\end{align}
where $\mathbb{P}_0, \mathbb{P}_1$ are the distributions \tdblue{under} the null hypothesis, and \tdblue{under} the hypothesis with $S$ being the outliers, respectively. Consequently, the error exponent achievable by $\delta'$ is the same as that achievable by $\delta$ for any $\mu, \pi$, $\mu \neq \pi$.

Having shown that the empirical distributions of the $M$ \tmblue{sequences} and $\pi$ are sufficient statistics, it suffices to consider tests that depend only on $\gamma_1, \ldots, \gamma_M$, and $\pi$. In particular, \tdblue{for any such $\delta$, let assume that for any $\pi$, there exists $\epsilon \,=\,\epsilon (\pi)\, > \, 0$ such that}
\begin{align}
 \lim\limits_{n \rightarrow \infty}
- \frac{1}{n} \log{ \mathbb{P}_0 \left \{ \delta \neq 0 \right \}} \ > \ \epsilon.
\label{PfProp2-1-1}
\end{align}
\tmblue{Let $E$ be the set of empirical distributions
\begin{align}
\hspace{-0.03in} E \triangleq 
\Big \{ \hspace{-0.02in} \left( q_1, \ldots, q_M \right) : \hspace{-0.01in} \sum\limits_{i \in S} D \left( q_i \| \pi \right) + \hspace{-0.02in} \sum\limits_{j \notin S} D \left( q_j \| \pi \right)
 \leq  \frac{\epsilon}{2} \Big \}.
\nonumber
\end{align} 
For an arbitrary element $\left( q_1, \ldots, q_M \right) \in E$, consider the set $A$ of all $M$ tuples $\left( \by^{(1)}, \ldots, \by^{(M)} \right) \in \mathcal{Y}^{Mn}$ conforming to the empirical distributions $\left( q_1, \ldots, q_M \right)$. It then follows from Lemma \ref{lm-1} that
\begin{align}
\lim_{n \rightarrow \infty} - \frac{1}{n} \log \mathbb{P}_{0} \Big \{ \big (\by^{(1)}, \ldots, \by^{(M)} \big ) \in A \} \nonumber \\
= \sum\limits_{i \in S} D \left( q_i \| \pi \right) \,+\, \sum\limits_{j \notin S} D \left( q_j \| \pi \right) \ \leq\  \frac{\epsilon}{2}. \nonumber
\end{align}
It now follows from (\ref{PfProp2-1-1}) that 
\begin{align}
E \subseteq \left \{ \delta = 0 \right \}. \nonumber
\end{align}
}
By applying Lemma \ref{lm-1} again, but now with respect to the hypothesis with $S$ being the outliers, we get that
\begin{align}
 & \hspace{-0.08in} \lim\limits_{n \rightarrow \infty}
- \frac{1}{n} \log \mathbb{P}_1 \left \{ \delta \neq 1 \right \} \nonumber \\
&  \hspace{-0.08in} \leq 
\min\limits_{\left( q_1, \ldots, q_M \right) \in E}\ \ 
\sum_{i \in S} D \left( q_i \| \mu \right) \ +\ 
\sum\limits_{j \notin S} D \left( q_j \| \pi \right).
\label{PfProp2-2-1}
\end{align}
Since $\epsilon$ is independent of $\mu$, and $\mu$ can be chosen arbitrarily close to $\pi,$ we can pick $\mu$ to be such that $\sum_{i \in S} D\left( \mu \| \pi \right) < \frac{\epsilon}{2}.$  It now follows from the definition of $E$, (\ref{PfProp2-2-1}) and Lemma \ref{lm-1} that
\begin{align}
\lim\limits_{n \rightarrow \infty}
- \frac{1}{n} \log{ \mathbb{P}_1 \left \{ \delta \neq 1 \right \}} \ =\ 0, \nonumber
\end{align}
\tdblue{which establishes the assertion, since if (\ref{PfProp2-1-1}) did not hold, the error exponent for $\delta$ would also have been zero.}

\subsection{Proof of Theorem \ref{thm-11}}

\tdblue{Without loss of generality, we can consider the following two hypotheses. The first hypothesis has $S_{1}$ as the set of outliers, and the second hypothesis has $S_{2},$ where $S_{1} \subset S_{2}$.  It suffices to prove that even when $\pi$ and $\{\mu_{i}\}_{i \in S_{1}}$ are known, there cannot exist a universally exponentially consistent test in differentiating such two hypotheses.  By the same argument as in the proof of Theorem \ref{thm-10}, we can consider tests that depend only on the empirical distributions of all the \tmblue{sequences} 
$\gamma_1, \ldots, \gamma_M$, $\pi$ and $\{ \mu_{i}\}_{i \in S_{1}}$.} 
Such a test is based on a decision rule $\delta : \cY^{Mn} \rightarrow \{1, 2\}$,  where $1$ corresponds to the hypothesis with $S_{1}$ being the outliers and $2$ to the hypothesis with $S_{2}$. In particular, \tdblue{let assume that for any fixed $\pi$ and $\{\mu_{i}\}_{i \in S_{1}}$, there exists $\epsilon \,=\,\epsilon \left (\pi, \{\mu_{i}\}_{i \in S_{1}} \right ) \, > \, 0$ such that}
\begin{align}
\lim\limits_{n \rightarrow \infty}
- \frac{1}{n} \log{ \mathbb{P}_1 \left \{ \delta \neq 1 \right \}} \ > \ \epsilon,
\label{PfProp2-1}
\end{align}
where $\mathbb{P}_1$ is the distribution \tdblue{under the hypothesis} with $S_{1}$ being the outliers.
It now follows from (\ref{PfProp2-1}) and Lemma \ref{lm-1} that the set \tdblue{$A$ of all $M$ tuples $\left( \by^{(1)}, \ldots, \by^{(M)} \right) \in \mathcal{Y}^{Mn}$ whose empirical distributions $\left( \gamma_1, \ldots, \gamma_M \right)$ lie in the following set}
\begin{align}
\hspace{-0.03in}E \triangleq \Big \{ \hspace{-0.03in} \left( q_1, \ldots, q_M \right) : \hspace{-0.03in} \sum\limits_{i \in S_{1}} \hspace{-0.03in} D \left( q_i \| \mu_{i} \right) + \hspace{-0.03in} \sum\limits_{j \notin {S_{1}}} \hspace{-0.03in}D \left( q_i \| \pi \right) \leq \frac{\epsilon}{2} \Big \}
\end{align} 
must be such that 
\begin{align}
A \subseteq \left \{ \delta = 1 \right \}.
\end{align}
By applying Lemma \ref{lm-1} again, but now with respect to the hypothesis with $S_{2}$ being the outliers, we get that
\begin{align}
&\lim\limits_{n \rightarrow \infty}
- \frac{1}{n} \log{ \mathbb{P}_2 \left \{ \delta \neq 2 \right \}} \nonumber \\
& \leq \ 
\min\limits_{\left( q_1, \ldots, q_M \right) \in E}\ 
\sum_{i \in S_{1}} D \left( q_i \| \mu_{i} \right) \, +\, \sum_{j \in S_{2}\setminus S_{1}} D \left( q_j \| \mu_{j} \right) \nonumber \\
&\hspace{1.5in} + \,
\sum\limits_{k \notin {S_{2}}} D \left( q_k \| \pi \right),
\label{PfProp2-2}
\end{align}
where $\mathbb{P}_2$ is the distribution \tdblue{under the hypothesis} with $S_{2}$ being the outliers.
Since $\epsilon$ is independent of $\{\mu_{j}\}_{j \in S_{2} \setminus S_{1}}$, we can pick $\{\mu_{j}\}_{j \in S_{2} \setminus S_{1}}$ to be such that $\sum\limits_{j \in S_{2}\setminus S_{1}}D\left( \mu_{j} \| \pi \right) < \frac{\epsilon}{2}.$  It now follows from the definition of $E$, (\ref{PfProp2-2}) and Lemma \ref{lm-1} that
\begin{align}
\lim\limits_{n \rightarrow \infty}
- \frac{1}{n} \log{ \mathbb{P}_2 \left \{ \delta \neq 2 \right \}} \ =\ 0, \nonumber
\end{align}
\tdblue{which establishes the assertion.}

\section{Optimal Test for One Outlier When Only $\mu$ Is Known}
\label{app-B}
Now we address the issue raised in Remark \ref{rmk-1}. In particular, when only $\mu$ is known, instead of using the corresponding version of \tmblue{the GL test in Section \ref{sec-UOHT-exactone-test}}, we adopt the following test $\tilde{\delta}$: 
\begin{align}
\tilde{\delta}(y^{Mn}) \ =\ \mathop{\arg\min}\limits_{i=1, \ldots, M} D(\gamma_{i} \| \mu), 
\label{eqn-scddetector}
\end{align}
where $\gamma_{i}$ denotes the empirical distribution of $\by^{(i)}, i=1, \ldots, M,$ and the ties in (\ref{eqn-scddetector}) are broken arbitrarily.

It now follows using (\ref{eqn-scddetector}), that 
\begin{align}
\mathbb{P}_{1} \{ \tilde{\delta} \neq 1 \} \ \leq \ (M-1) \mathbb{P}_{1} \big \{ D(\gamma_{1} \| \mu) \ \geq \ D(\gamma_{2} \| \mu) \big \} \nonumber
\end{align}
Applying Lemma \ref{lm-1} with $J=2$, $p_{1}=\mu$, $p_{2}=\pi$, and
\begin{align}
E \ =\ \big \{ \left(q_{1}, q_{2} \right) ~:~ 
D \left( q_{1} \| \mu \right) \ \geq \
D \left( q_{2} \| \mu \right) \big \}, \nonumber
\end{align}
we get that the exponent for $\mathbb{P}_{1} \{ \tilde{\delta} \neq 1 \}$ is given by the value of the following optimization problem
\begin{align}
&\mathop{\min\limits_{q_{1}, q_{2} \,\in \, \mathcal{P}(\cY)}}_{\scriptsize{{D \left ( q_{1} \| \mu \right )  
\ \geq\ D \left ( q_{2} \| \mu \right )}} }   D \left ( q_{1} \| \mu \right ) + D \left ( q_{2} \| \pi \right )
\hspace{1.0in} \nonumber \\
& \hspace{0.4in}\geq\ \min_{q_{2}} \ D \left ( q_{2} \| \mu \right ) + D \left ( q_{2} \| \pi \right ), \nonumber \\
&\hspace{0.4in} =\ \ 2 B(\mu, \pi). \nonumber
\end{align}
\tdblue{where the inequality follows by substituting the constraint into the objective function and the equality follows from Lemma \ref{lm-2}.}

\section*{Acknowledgments}

The authors thank \tdblue{Dr. Jean-Francois Chamberland-Tremblay, Aly El Gamal and Dr. Maxim Raginsky} for interesting discussions related to this paper.

\bibliographystyle{IEEEtran}
\bibliography{Generalized-UOD}

\bigskip
Yun Li (SM'13) received the B.Sc. degree from Shanghai Jiao Tong University, China, in 2009, and the M.Sc. degree from the University of Illinois, Urbana-Champaign, in 2012. She is pursuing a Ph.D. degree in the Department of Electrical and Computer Engineering at the University of Illinois, Urbana-Champaign. Her research interests include detection and estimation theory,
statistical learning, and information theory.
\bigskip

Sirin Nitinawarat (SM'09--M'11) obtained the B.S.E.E. degree from Chulalongkorn University, Bangkok, Thailand, with first class honors, and the M.S.E.E. degree from the University of Wisconsin, Madison.  He received his Ph.D. degree from the Department of Electrical and Computer Engineering and the Institute for Systems Research at the University of Maryland, College Park, in December 2010.  He is now a postdoctoral research associate at the Coordinated Science Laboratory at the University of Illinois at Urbana-Champaign.  His research interests are in statistical signal processing, estimation and detection, information and coding theory, communications, stochastic control, and machine learning.
\bigskip

Venugopal V. Veeravalli (M'92--SM'98--F'06) received the B.Tech. degree (Silver Medal Honors) from the Indian Institute of Technology, Bombay, in 1985, the M.S. degree from Carnegie Mellon University, Pittsburgh, PA, in 1987, and the Ph.D. degree from the University of Illinois at Urbana-Champaign, in 1992, all in electrical engineering.

He joined the University of Illinois at Urbana-Champaign in 2000, where he is currently a Professor in the department of Electrical and Computer Engineering and the Coordinated Science Laboratory. He served as a Program Director for communications research at the U.S. National Science Foundation in Arlington, VA from 2003 to 2005. He has previously held academic positions at Harvard University, Rice University, and Cornell University, and has been on sabbatical at MIT, IISc Bangalore, and Qualcomm, Inc. His research interests include wireless communications, distributed sensor systems and networks, detection and estimation theory, and information theory.

Prof. Veeravalli was a Distinguished Lecturer for the IEEE Signal Processing Society during 2010--2011. He has been on the Board of Governors of the IEEE Information Theory Society. He has been an Associate Editor for Detection and Estimation for the IEEE Transactions 
\newpage
on Information Theory and for the IEEE Transactions on Wireless Communications. Among the awards he has received for research and teaching are the IEEE Browder J. Thompson Best Paper Award, the National Science Foundation CAREER Award, and the Presidential Early Career Award for Scientists and Engineers (PECASE).

\end{document}